\documentclass[runningheads]{llncs}
\usepackage{cite}
\usepackage{pdfpages}
\usepackage{amsmath,amssymb,amsfonts}
\usepackage{algorithmic}
\usepackage{graphicx}

\usepackage{textcomp}
\usepackage{xcolor}
\usepackage{pdflscape}

\def\BibTeX{{\rm B\kern-.05em{\sc i\kern-.025em b}\kern-.08em
    T\kern-.1667em\lower.7ex\hbox{E}\kern-.125emX}}
    
\usepackage{subfiles}
\usepackage{tikz} 
\usepackage{pbox}

\usepackage{todonotes}

\usepackage{multirow}
\usepackage{diagbox}

\usepackage{amsmath}

\newcommand{\meta}{\text{X}}
\newcommand{\IR}{\mathbf{R}}

\newcommand{\IB}{\mathbf{B}}
\newcommand{\BX}{\mathbf{B}_\meta}

\newcommand{\calI}{{\cal I}}
\newcommand{\calL}{{\cal L}}
\newcommand{\calO}{{\cal O}}

\renewcommand{\paragraph}[1]{{\medskip\noindent\bf #1.}}
\newtheorem{thm}{Theorem}
\newtheorem{lem}{Lemma}

\begin{document}

\title{On the Susceptibility of QDI Circuits to\\Transient Faults\thanks{This research was partially supported by the project ENROL (grant I 3485-N31) of the Austrian Science Fund (FWF), the Doctoral College on Resilient Embedded Systems (DC-RES), the ANR project DREAMY (ANR-21-CE48-0003), and the French government's excellence scholarships for research visits.}
}


\author{Raghda El Shehaby\inst{1}\orcidID{0009-0000-6653-9074}
\and
Matthias F\"ugger\inst{2}\orcidID{0000-0001-5765-0301}
\and
Andreas Steininger\inst{1}\orcidID{0000-0002-3847-1647}
}

\authorrunning{El Shehaby \emph{et al.}}

\institute{
  TU Wien, Institute of Computer Engineering\\
  \and
  CNRS \& LMF, ENS Paris-Saclay, Université Paris-Saclay \& Inria
}

\maketitle

\begin{abstract}
By design, quasi delay-insensitive (QDI) circuits exhibit hi\-gher resilience against timing variations as compared to their synchronous counterparts.
Since computation in QDI circuits is event-based rather than clock-triggered, spurious events due to transient faults such as radiation-induced glitches, a priori are of higher concern in QDI circuits.

In this work we propose a formal framework with the goal to gain a deeper understanding on how susceptible QDI circuits are to transient faults.
We introduce a worst-case model for transients in circuits.
We then prove an equivalence of faults within this framework and use this result to provably exhaustively check a widely used QDI circuit, a linear Muller pipeline, for its susceptibility to produce non-stable output signals.


\end{abstract}

\section{Introduction}
\label{sec:intro}

A transient fault in a circuit is a temporarily incorrect value at a circuit's signal, e.g., induced by radiation.
It is well known that synchronous (i.e., clocked) circuits exhibit a natural resilience against transient faults through masking. Specifically, the relevant effects are electrical masking (short fault pulses are filtered by low-pass behavior of subsequent gates and interconnect), logical masking (depending on other input levels, the logic level of the faulty input may be irrelevant for the gate output) and temporal masking (the flip-flop samples its data input at the active clock edges while ignoring faults that happen between these). However, synchronous circuits have little resilience against (fault) effects that impact the timing.
By contrast, asynchronous (i.e., self-timed, handshake-based) and in particular quasi delay-insensitive (QDI) \cite{martin1986compiling} circuits exhibit large, ideally unlimited, tolerance against timing variations by construction. This is due to their event-driven operation principle.
Unfortunately, this very event driven operation makes them prone to transient faults. While electrical masking and logical masking mitigate fault effects like in the synchronous case, it is not obvious whether considerable temporal masking occurs. Previous works have shown that asynchronous pipelines, e.g., have data accepting windows during which they are susceptible to fault pulses. The size of these windows depends on several parameters, most notably the mode of pipeline operation (bubble-limited/balanced/token-limited). For unbalanced operation these windows may reach considerable size, making the circuit clearly more susceptible to faults than in the synchronous case with its instantaneous sampling. That is why several mitigation methods \cite{bainbridge2009glitch} aim at minimizing the data accepting windows. In any case there is some effect equivalent to temporal masking, and most often it is constituted by Muller C-elements (MCEs): During the \emph{combinational} mode of operation (matching inputs), the MCE ignores fault pulses on any input and not even a pulse at the output can flip its state. In \emph{storage} mode (non-matching inputs), however, the MCE's state can be easily flipped by a fault pulse at one of the inputs or at the output (directly at the keeper). So apparently, the share of time during which an MCE is in combinational mode determines the masking provided by it. In a reasonably complex practical setting, however, this insight is hard to map to a general prediction of the whole circuit.

Given an asynchronous circuit, a natural question thus is at which signals and at which times the circuit is susceptible to a transient fault.
In this paper we present an approach to efficiently and provably exhaustively answer this question. 

\paragraph{Organization}
We discuss related work in Section~\ref{sec:relatedwork} and introduce our circuit model in Section~\ref{sec:model}.
In Section~\ref{sec:results} we start with basic consistency results
  of the model, followed by our main technical result: the definition of value regions in executions along with a proof of the equivalence of glitches within those regions (Theorem~\ref{thm:main}).
Based on this result we then present our tool for sensitivity-window exploration (Section~\ref{sec:tool}) and apply it to
a  widely used QDI circuit for illustration.
We conclude in Section~\ref{sec:conclusion}.


\section{Related Work}
\label{sec:relatedwork}

\paragraph{Transient faults in asynchronous circuits}
Several studies have explored the effects of transient faults on asynchronous circuits. Detection and mitigation techniques with some form of redundancy have been proposed alongside.

The authors in~\cite{lafrieda2004fault} perform a thorough analysis of single-event transient (SET) effects, among other types of faults, in QDI circuits. The fault's impact is first presented at the gate level, then on communication channels, translating the fault to a deadlock. They also discuss other possible errors (synchronization failure, token generation, and token consumption).
An efficient failure detection method for QDI circuits is presented in~\cite{peng2005efficient}. The method brings the circuit to a fail-safe state in the presence of hard and soft errors. The authors investigate the probability for a glitch to propagate through a state-holding element in asynchronous circuits.
In~\cite{monnet2007formal}, the authors propose a formal method to model the behavior of QDI circuits in the presence of transient faults. They use symbolic simulation to provide an exhaustive list of possible effects and analyze which of these cases are theoretically reachable. Their model, however, does not support delay parameters, which potentially reduces the set of reachable states, further improving the resistance of a design against single-event upsets (SEUs).
They also discuss in~\cite{monnet2005asynchronous} the Muller C-element fault sensitivity and specify a global sensitivity criterion to SETs for asynchronous circuits. The work provides a behavioral analysis of QDI circuits in the presence of faults.
With the help of signal transition graphs (STGs), the authors in~\cite{bainbridge2009glitch} informally analyze SEUs due to glitches on QDI network-on-chip links. Several mitigation techniques are proposed with a focus on reducing the latch's sensitive window to a glitch.
Some of these techniques are tested and compared against other proposed variations in \cite{huemer2020QDIwindows}, \cite{behal2021towards}, and \cite{tabassam2022set}. The assessment there is based on extensive fault injection simulations into different QDI buffer styles, in order to identify the main culprits of the circuit. The authors provide a quantitative analysis to determine the windows of vulnerability to SETs and the impact of certain parameter choices on the resilience of the circuit. However, the analysis is done based on a regular timing grid, which causes linear complexity in time and in resolution, and cannot exclude the potential of overlooking relevant windows between the grid points.

\paragraph{Hazards in PRSs}
QDI circuits can be modeled on different levels. The Production Rule Set (PRS), introduced by Martin~\cite{martin1986compiling}, is a widely-used low-level representation that can be directly translated to a CMOS transistor implementation. PRSs do not normally support hazards, and by guaranteeing \emph{stability} and \emph{non-interference} characteristics~\cite{jang2005seu}, a PRS execution is assumed to be hazard-free. The authors consider an SEU as flipping of a variable's value and model it in so called transition graphs to identify deadlock or abnormal behavior.
\cite{katelman2012rewriting} extends the semantics of PRSs in order to be able to address hazards as circuit failures, but it is limited to checking the hazard-freedom property of a circuit. 

These papers are focused on the \emph{possibility} of failure and are restricted to precedence of events, without explicitly considering timing.
Our work enables further propagation of what we define as a glitch in order to check whether it has reached the final outputs of a circuit and, based on actual timing information, \emph{quantify} this proportion of failure.


\section{Model}
\label{sec:model}

Following the work by Martin \cite{martin1986compiling}, we model a circuit as a set of production rules.
We extend the model by delays and propagation of non-Boolean values.
We start with definitions of signal values and production rules in our context. 

\paragraph{Signal and signal values}
Signals are from a finite alphabet $S$.
Signals have values that may change over time.
We extend the values a signal may attain from the classical Boolean values $\IB = \{0,1\}$ to the three-valued set $\BX = \{0,\meta,1\}$, where $\meta$ is a potentially non-binary value.
Examples for non-binary values are glitches, oscillations, and metastable values.
A signal that has value $\meta$ may, however, be $0$ or $1$.

We will make use of logical operations like $\wedge$ and $\neg$ on the extended domain $\BX$.
If not stated otherwise, we resort to the semantics of the 3-valued Kleene logic, introduced by Goto for these operations; see \cite{brzozowski2001algebras}.
In short, using the classical algebraic interpretation of Boolean formulas on $\{0,1\} \subset \IR^+_0$  where, $\neg a \equiv 1 - a$, $a \wedge b \equiv \min(a,b)$, and $a \vee b \equiv \max(a,b)$, one obtains the Kleene semantics by the correspondence $\meta \equiv 1/2$.
For example, one obtains, $1 \wedge \meta = \meta$ and $1 \vee \meta = 1$.

\paragraph{Production rules}
A production rule is a guarded Boolean action with delay.
It is of the form
\begin{align}
G \rightarrow s=1 \ [d] \quad\text{ or }\quad G \rightarrow s=0 \ [d]\enspace,
\end{align}
where the guard $G$ is a logical predicate on signals, $s$ is a signal, and $d \in (0,\infty)$ is the propagation delay.
Intuitively, a production rule with guard $G$,
 action $s = b$, where $b \in \{0,1\}$,
 and delay $d$ sets signal $s$'s value to $b$ upon predicate $G$ being true for $d$ time.

\paragraph{Circuit}
A circuit is specified by:
\begin{itemize}
\item Finite, disjoint sets of input, local, and output signals, denoted by $\calI$, $\calL$, and $\calO$.

\item Initial values for all local and output signals.
  We write $s(0)$ for the initial value of signal $s \in \calL \cup \calO$.

\item A set of production rules $R$ whose guards are predicates on the circuit's signals and whose actions involve only local and output signals.
We require that (i) for each signal $s$, there is at most one production rule that sets $s$ to $1$, and at most one that sets $s$ to $0$,
and (ii) guards of production rules that set a signal $s$ are mutually exclusive for all signal values from $\IB$. 
\end{itemize}

Similarly to Martin \cite{martin1986compiling} we use production rules to model gates: actions that set a value to $1$ correspond to the pull-up stack of a gate and actions that set a value to $0$ to the pull-down stack.
Any meaningful circuit will further have the properties that any local and output signal appears in a production rule that sets it to $0$ and one that sets it to $1$; if not, the signal will remain at its initial value for all times.
Further, as already demanded in the last bullet above, the guards of these opposing production rules will not both evaluate to true for any choice of signal values; if not, the pull-up and pull-down stacks of this gate will drive the gate's output at the same time.

\paragraph{Signal trace}
A signal trace for signal $s \in S$ is a function $v_s: \IR_0^+ \to \BX$
  mapping the time $t$ to the value of $s$ at time $t$.
By slight abuse of notation, we write $s(t)$ for $v_s(t)$.
We restrict signal traces to contain only finitely many value-changes in each finite time interval.

\paragraph{Execution}
It remains to define how a circuit, with a given input, switches signal values.
For that purpose fix a circuit, input signal traces for all its inputs~$I$, and a time~$T > 0$ until which the execution is to be generated.

Intuitively an \emph{execution induced by the circuit and the input signal traces} is inductively generated via applying the production rules to the current signal values.
If a guard of a production rule is true, its action is scheduled to take place after the rule's delay.

Care has to be taken to handle instability of guards.
If a guard that results in a scheduled action on a signal, but whose action has not yet been applied, becomes false, we remove the scheduled action and instead set the signal to $\meta$ after a small delay $\varepsilon > 0$.
An $\varepsilon$ smaller than the rule's delay
accounts for the fact that non-binary outputs can propagate faster than full-swing transitions.
The signal's value $\meta$ is then propagated accordingly throughout the circuit.
Indeed we will let $\varepsilon \to 0$ in later sections to account for the
  worst case behavior of gates.

Formally, the \emph{execution prefix until time $T$, induced by the circuit and the input signal traces}, is a signal trace prefix until time $T$ for each local and output signal obtained as follows:
\begin{enumerate}
\item Initially, all signals are set to their initial values as specified by the circuit.
Further, the current time $t = 0$, and the set of scheduled actions is empty.

\item Handle unstable guards:
\begin{itemize}
\item For each production rule whose action $s = b$, with $b \in \IB$, currently being scheduled: if the rule's guard evaluates to $0$ or $\meta$, and $s(t) \neq b$ (we say the guard is unstable), then remove the event from the scheduled events and set $s = \meta$. \emph{(generate-$\meta$)}
\end{itemize}

\item Apply actions:
\begin{itemize}
\item For each action $s = v$, with $v \in \BX$, scheduled for time $t$, set $s(t) = v$ and remove the action from the scheduled actions.
\end{itemize}

\item Schedule actions:
\begin{itemize}
\item For each production rule: if its guard evaluates to $1$, schedule the rule's action $s = b$ to take place after the rule's delay $d$, i.e., at time $t + d$ (unless $s(t) = b$ already).

\item For each production rule: if its guard evaluates to $\meta$ and the rule's action is $s = b$ with $s(t) \neq b$, schedule the action $s = \meta$ for time $t + \varepsilon$ (unless $s(t) = \meta$ already). \emph{(propagate-$\meta$)}
\end{itemize}

\item Advance time $t$ to the nearest future time at which an action is scheduled or an input signal switches value.
If $t \geq T$, return the local and output signal traces until time $T$; otherwise, continue with step 2.
\end{enumerate}

One observes that an execution prefix until time $T' > T$ is an extension of an execution prefix until time $T$: for each local and output signal $s$, the signal values in both prefixes are identical within $[0,T]$.
We may thus speak of the execution as the limit of execution prefixes until times $T \to \infty$.

\subsection{Example}
As an example let us consider the circuit with input signal $i$, no local signals, and output signal $o$.
As initial value we choose $o(0) = 1$.
The circuit comprises of a single inverter with input $i$, output $o$, and delay $1.0$, i.e., the circuit's production rules are:
\begin{align}
i &\rightarrow o=0 \ [1.0] \label{eq:Rd}\\
\neg i &\rightarrow o=1 \ [1.0]\enspace \label{eq:Ru}.
\end{align}

We consider three input traces: (a) Initially $i(0) = 0$, then $i$
  transitions to $1$ at time $1$ where it remains.
  (b) Prefix like (a), but the input transitions back to $0$ at time $1.5$.
  (c) Like (b), but with value $\meta$ during times $[1,1.5)$.
  
The execution prefixes until time $T = 4$ induced by the above
  circuit and the input signal traces (a), (b), and (c) are depicted
  in Figure~\ref{fig:ex}.

\begin{figure}[h]
    \centering
    \includegraphics[width=0.2\columnwidth]{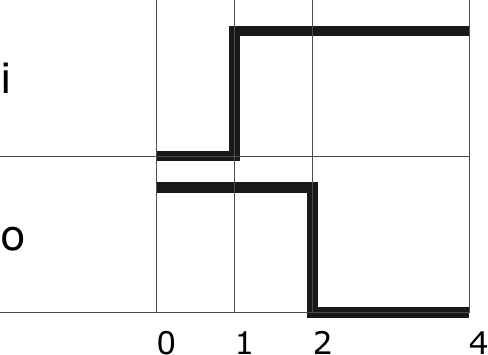}
    ~
    \includegraphics[width=0.2\columnwidth]{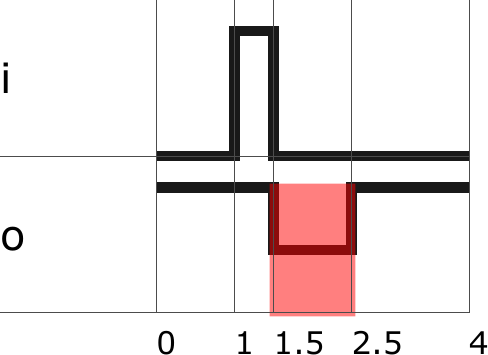}
    ~
    \includegraphics[width=0.2\columnwidth]{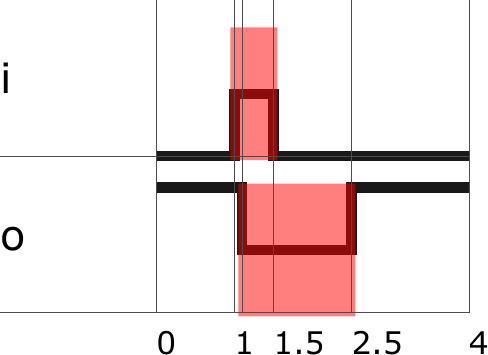}
    \caption{Execution prefixes until time $T=4$ of an inverter with input {\tt i} and output {\tt o}. 
    Signal value $\meta$ is depicted as a value of $0.5$ and marked red.
    The propagation delay $\varepsilon$ for signal value $\meta$ is set to $0.1$. Left: input signal trace (a). Middle: input signal trace (b). Right: input signal trace (c).}
    \label{fig:ex}
\end{figure}

In the example, input traces~(a) and (b) result in the guard of rule~\eqref{eq:Rd}
  becoming true at time $1$.
Accordingly, an action to set $o = 0$ is scheduled for time $1 + d = 2$.
While in input trace~(a), the guard remains true until time $2$,
  and thus $o$ is set to $0$ at time $2$,
  in input trace~(b), the guard is falsified at time~$1.5$,
  resulting in the action being canceled and $o$ is set to $\meta$ at time~$1.5$ (generate-X in the algorithm).
  
For input trace~(b), we have that the guard of rule~\eqref{eq:Ru} becomes
  true at time $1.5$.
Accordingly the action $o = 1$ is scheduled for time $1.5 + d = 2.5$.
Since the guard remains true until time $2.5$, the action is applied
  resulting in $o(2.5) = 1$.

Finally, input trace (c) demonstrates the algorithmic rule propagate-X in step 5: the $\meta$ value at the input is propagated with propagation delay $\varepsilon = 0.1$ to the output.
Resetting the output to $1$ at time $2.5$ occurs as for input trace (b).


\section{Results}
\label{sec:results}

\subsection{Well-defined executions}

We start with a basic result on the consistency of an execution as defined by the algorithm.

\begin{lem}\label{lem:welldefined}
Any signal trace of an execution has at most finitely many value-changes within a finite time interval; it is thus a well-defined signal trace.
\end{lem}
\begin{proof}
Assume by contradiction that a signal trace has infinitely many value-changes within a finite interval $[t,t'] \subset \IR$.
By consistency of prefixes of executions, this implies that the algorithm returns
  an execution with infinitely many value-changes when setting $T = t'$.

In the algorithm, at any point in time $\tau$ there is at most one action per non-input signal in the set of scheduled actions and at most bounded many actions per input signal until time $T$.
Observing that there is a minimum propagation delay $d_{\min} > 0$ for signal values $0$, $1$, and $\meta$, any newly scheduled action must occur at earliest at time $\tau + d_{\min}$.
Thus, only bounded many actions occur within $[\tau,\tau + d_{\min}]$.
The statement follows.\qed
\end{proof}

\subsection{A transient-fault insertion tool}

To study the effect of short transient faults on the behavior of circuits we extend the algorithm from Section~\ref{sec:model} to allow the insertion of external events: signal transitions from a set of \emph{external events} are applied at the end of step~3.
Step~5 is changed to include external events when updating time~$t$ to the time of the next event.
A transient fault then corresponds to two subsequent signal transitions of the same signal in the set of external events in our model.
This is less general than transient faults in physical implementations, where a transient fault, e.g., induced by an additional charge due to a particle hit, can lead to a single early transition by happening just before a valid signal transition.
The assumption, however, is conservative in the sense that we assume that such a charge is small enough to lead to a pulse, i.e., double transition, potentially violating a gate's stability condition: in fact we will later in Section \ref{sec:equivalence-transient-faults} assume that transient faults are not necessarily full-swing binary pulses and have value $\meta$ in our model.

We have implemented the algorithm in Python and shall discuss results for a widely-used QDI circuit component, a linear pipeline, in the following.

\paragraph{Linear pipeline}
To study the susceptibility of QDI circuits to transient faults, we used the tool to insert short pulses (glitches) at different times.
As a prototypical QDI circuit, we used the linear 3-stage Muller pipeline shown in Figure~\ref{fig:muller3_linear}. 

\begin{figure}[h]
    \centering
    \includegraphics[width=0.6\columnwidth]{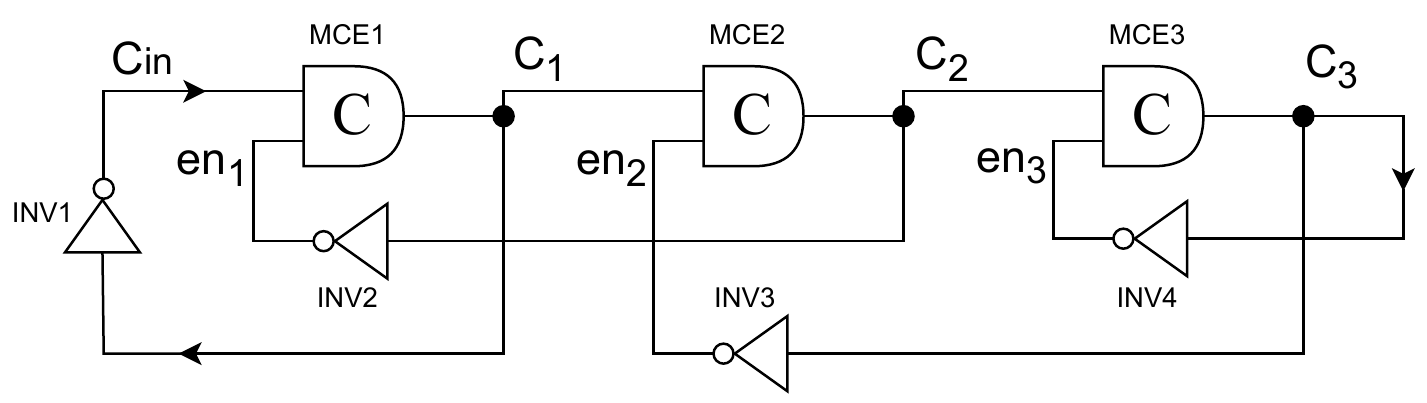}
    \caption{Linear Muller pipeline with 3 stages. The delays are set to 1 ({\tt INV2, INV3}), 5 (C gate), 4 (source delay = {\tt INV1}), and 4 (sink delay = {\tt INV4}).}
    \label{fig:muller3_linear}
\end{figure}

Delays have been uniformly set to 1 for the two pipeline inverters {\tt INV2} and {\tt INV3}, to 5 for all Muller C-elements ({\tt MCE1} to {\tt MCE3}), and 4 for the leftmost inverter {\tt INV1} and the rightmost inverter {\tt INV4}, which model the source and sink of the pipeline, respectively. 
Figure~\ref{fig:trace} shows an execution prefix until time $T=32$ in absence of transient faults, generated by our tool.
\begin{figure}[h]
    \centering
    \includegraphics[width=0.7\columnwidth]{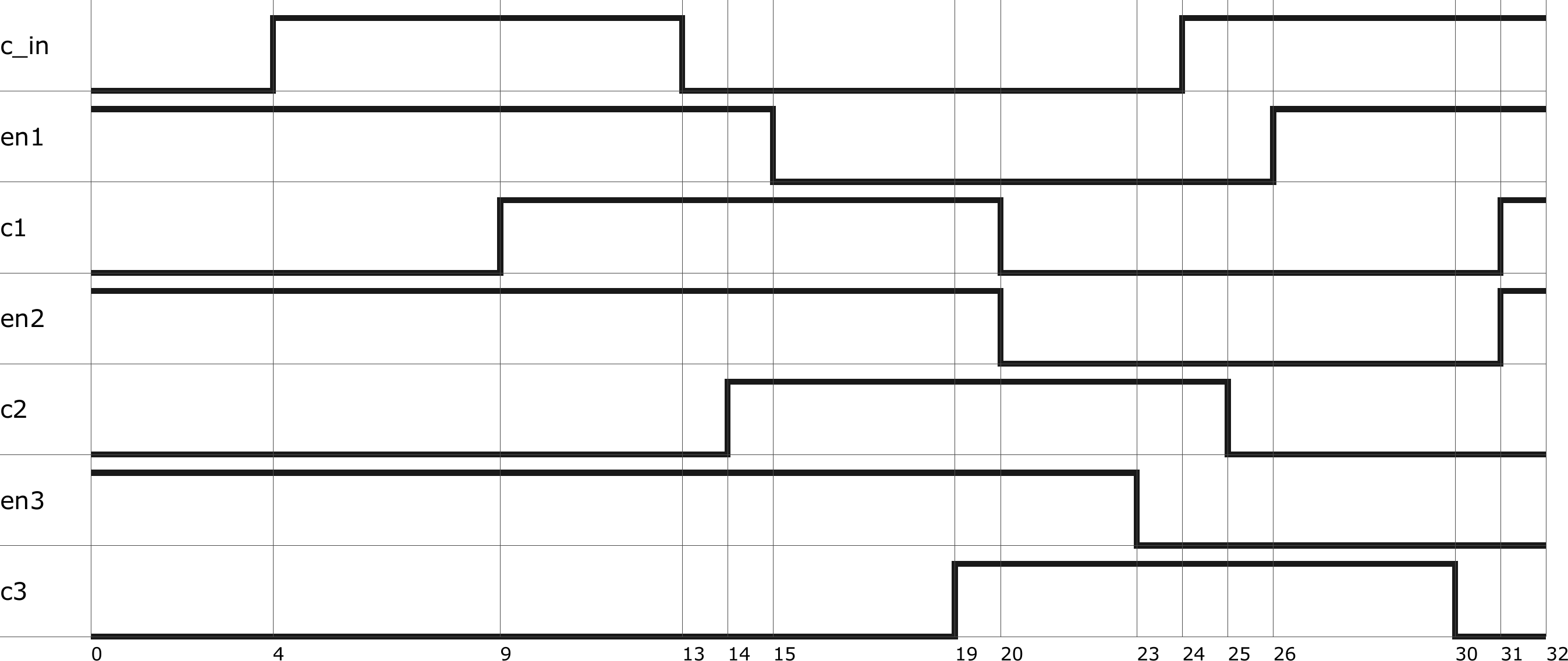}
    \caption{Execution prefix of linear 3-stage pipeline until time $T=32$.}
    \label{fig:trace}
\end{figure}
Figures~\ref{fig:trace-glitch} and \ref{fig:trace-glitch-ii} show execution prefixes of the same circuit until time $T=32$ when a glitch of width $0.1$ is inserted at the same signal, {\tt c2}, at different points in time:
  the intervals during which a signal has value $\meta$ are marked in red.
One observes that the behavior is different in presence of the glitch, as detailed in the following.

\begin{figure}[h]
    \centering
    \includegraphics[width=0.7\columnwidth]{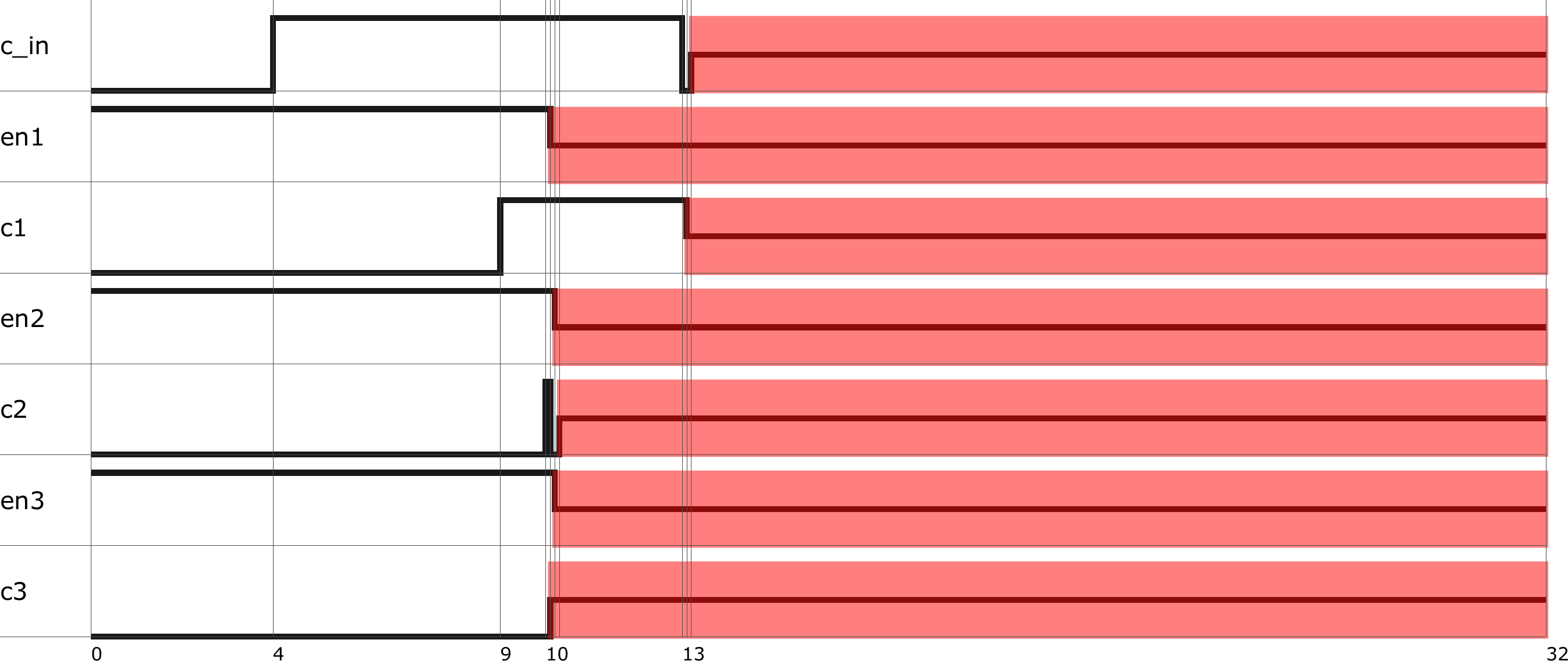}
    \caption{Execution prefix of linear 3-stage pipeline until time $T=32$ with glitch of width $0.1$ inserted at time $10$ at signal {\tt c2}.}
    \label{fig:trace-glitch}
\end{figure}

\emph{Non-masked glitch.} In Figure~\ref{fig:trace-glitch}, the glitch occurs at the input of the MCE while it is in storage mode, i.e., non-matching inputs. Since the other stable input {\tt en3} is at a different logic level than the MCE output {\tt c3}, the $\meta$ value is generated at the latter signal, and subsequently propagates through the circuit.

\begin{figure}[h]
    \centering
    \includegraphics[width=0.7\columnwidth]{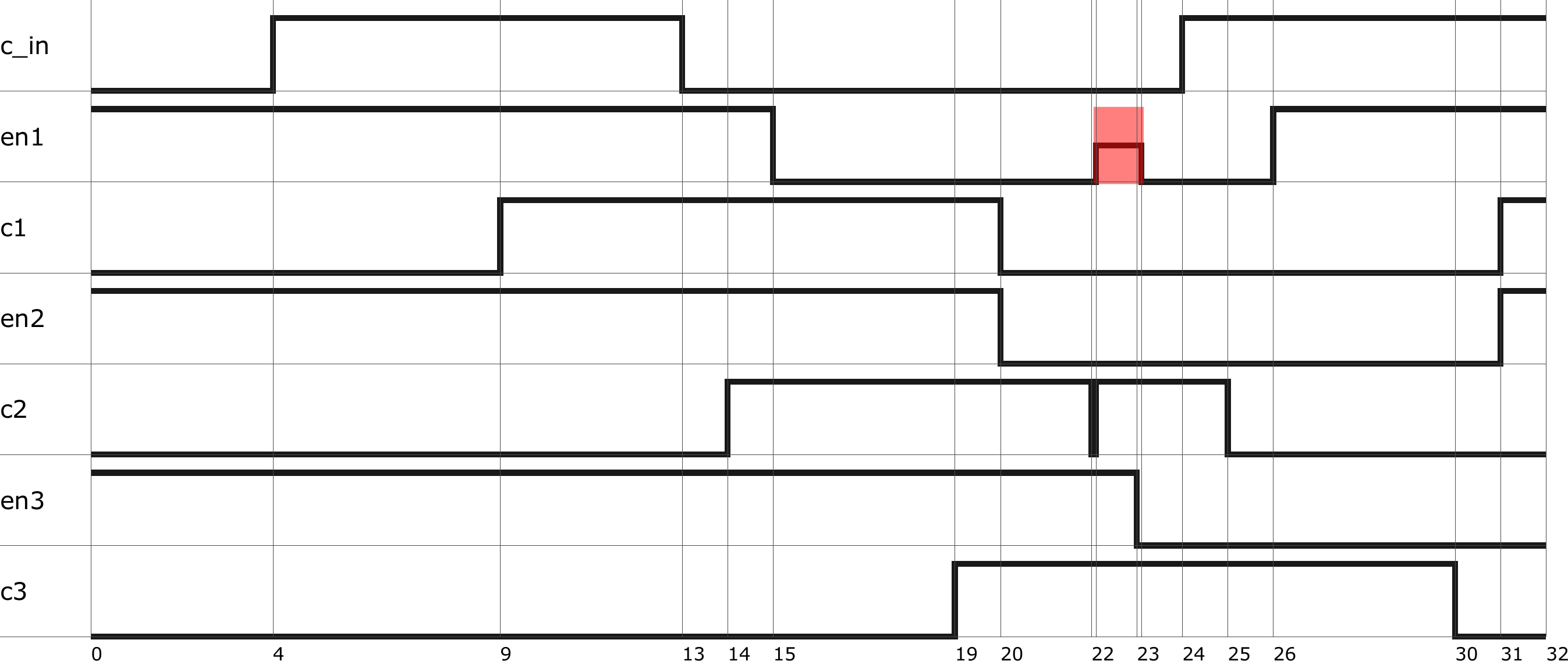}
    \caption{Execution prefix of linear 3-stage pipeline until time $T=32$ with glitch of width $0.1$ inserted at time $22$ at signal {\tt c2}.}
    \label{fig:trace-glitch-ii}
\end{figure}

\emph{Masked glitch.} By contrast, the glitch in Figure~\ref{fig:trace-glitch-ii} occurs at the MCE input while in combinational mode, i.e., matching inputs. The glitch is masked at the output {\tt c3}, but the $\meta$ value appears for a short period of time at {\tt en1} (since an inverter propagates its input value). During this time span, the $\meta$ value appeared and disappeared while the other MCE was also in combinational mode, hence was prevented from propagating the unstable value further in the circuit.

\paragraph{Susceptibility to transient faults}
The two different behaviors raise the question of when a QDI circuit like the linear pipeline can mask glitches successfully, and when it is susceptible to them.
To address that, we relate susceptibility to the occurrence of glitches at signals of particular interest (typically the output signals to the environment).
We call these signals of interest the \emph{monitored signals}.
For example, in the linear pipeline, signals {\tt c1} and {\tt c3} are the outputs to the environment represented by the source on the left and the sink on the right.

In general, let $C$ be a circuit, and $i$ an input signal trace.
Let $M \subseteq \calL \cup \calO$ be the set of monitored signals.
Then, $(C,i)$ is \emph{susceptible to a glitch (of width $w$) at signal $s \in \calI \cup \calL \cup \calO$ at time $t$}, if there exists a signal $m \in M$ and a time $t'$ such that in the execution, induced by the circuit $C$ and the input signal traces $i$ and with a glitch (of width $w$) at signal $s$ and time $t$, it is $m(t') = \meta$.

Revisiting the example of the linear pipeline, and letting $M = \{{\tt c1}, {\tt c3}\}$ be the set of monitored signals, the pipeline with its input is susceptible to a glitch (of width $0.1$) at signal {\tt c2} at time $10$, but not at time $22$ (see Figures \ref{fig:trace-glitch} and \ref{fig:trace-glitch-ii}).

This directly leads to the question of the sensitivity windows, i.e., the times when a circuit with an input is susceptible and when not.
Related, if combined with a probability measure on faults occurring at these times, one may ask how likely a transient fault is to cause a glitch at a monitored signal.
We address both questions in the following.

\subsection{Equivalence of transient faults}
\label{sec:equivalence-transient-faults}

While the previous tool allows one to sample the susceptibility at particular times, such an approach has several drawbacks: (i) it is time consuming to generate such sweeps and (ii) small susceptible windows may be overlooked.

In the following we present an alternative approach that relies on showing the equivalence between certain transient faults.
We begin the analysis with a definition.
We say \emph{signal $s$ has a pulse at time $t$ of width $w > 0$} if $s$ changes value at time $t$, remains constant within time $[t,t+w)$, and changes value at time $t+w$.
A $v$-pulse, with $v \in \BX$, is a pulse that has value $v$ during $[t,t+w)$.
We speak of a \emph{transient fault} as an $\meta$-pulse that has width of at most some predefined small $\gamma > 0$.

We are now ready to show a monotonicity property of the value $\meta$ in executions: If transient faults are added to a circuit's execution, the resulting execution differs from the original one at most by turning Boolean values into $\meta$.
For example, a transient fault may not result in a later 0-1 transition.

\begin{thm}[Monotonicity of $\meta$]\label{thm:monotonicity}
Let $C$ be a circuit and $i$ be input traces.
Let~$e$ be the execution induced by circuit~$C$
  and input traces~$i$,
  and $e'$ the execution induced by circuit $C$
  and input $i$ in presence of transient faults.
Then for all signals $s$ and times $t$,
if $s(t) \in \IB$ in~$e'$, then~$s(t)$
  is identical in~$e$ and~$e'$.
\end{thm}
\begin{proof}
Assume, by means of contradiction, that the statement does not hold and let $t$ be the smallest time at
which executions~$e$ and~$e'$ do not fulfill the theorem's statement.
Then there is a signal $s$ such that $s(t) = b \in \IB$
  in execution $e'$ and $s(t) = v \neq b$ in execution~$e$.
We distinguish between two cases for value~$v$:

\medskip

Case $v = \meta$.
If so, in execution $e$, signal $s$ was set to $\meta$
  at some time $\tau \leq t$, and not set again
  within $[\tau,t]$.
By minimality of $t$, and the definitions of $e$ and $e'$,
  $s$ was also set to $\meta$ in $e'$ at time $\tau$ (or earlier).
It follows that in execution $e'$,
  signal $s$ was set to $b$ within $(\tau,t]$.
This implies that a rule with guard $G$ and action $s = b$ was triggered at a time before $t$, and thus
$G$ was true in execution $e'$.
By minimality of $t$ and the definitions of $e$ and $e'$,
  $G$ must have been also true in $e$, resulting in the
  same action being scheduled also in $e$;
  a contradiction to the assumption that $v = \meta$.

\medskip

Case $v = \neg b$.
If so, $s$ was set via two different rules in~$e$ and~$e'$ and not set to another value until time $t$.
This implies that mutually exclusive guards have evaluated to $1$ in $e$ and $e'$ before time $t$; a contradiction to the minimality of $t$ in combination with the theorem's statement.

\medskip

The theorem's statement follows in both cases.\qed
\end{proof}

We next define time intervals that play a central role in a circuit's behavior in presence of transient faults.

Given a circuit $C$ and an execution $e$ of $C$,
the \emph{set of value switching times}, $V_{C}(e)$, is the set of times $\tau_0 = 0, \tau_1, \dots$ at which a signal in execution $e$ switches value.
A \emph{value region of execution $e$} is an interval $[t,t') \subset \IR$, where $t,t'$ are consecutive value switching times of execution $e$.
A \emph{postfix of a value region $[t,t')$} is a (potentially empty) interval $[t'',t') \subseteq [t,t')$.

\begin{thm}\label{thm:main}
Let $C$ be a circuit,
$i$ be input traces,
$\gamma > 0$ the width of a transient fault,
and $\varepsilon > 0$ the propagation delay of
  value $\meta$.
Let $e$ be the execution induced by circuit $C$
  and input traces~$i$.
  
Then, for a signal~$s$ of the circuit, 
  and a value region $R$ of execution $e$,
  the set $\Sigma_s(R)$ of times $t \in R$ such that $(C,i)$ is susceptible to a transient fault (of width $\gamma$) at signal $s$ at time $t \in R$ converges to a postfix of $R$ as $\varepsilon \to 0$ and $\gamma \to 0$.
\end{thm}

This means that every value region can be split into two intervals per signal: the left part of the region that contains the times at which the circuit
is not susceptible, and the right part where it is susceptible to faults.
Either part/interval can be empty.

\paragraph{Proof}
In the following fix $C$, $i$, $\gamma$, $\varepsilon$, execution $e$, signal $s$, and value region $R$ of execution $e$.
We first show a monotonicity property within a value region.

\begin{lem}\label{lem:tau1M-to-tau2M}
Let $R = [t,t')$ be a value region of execution $e$ and
  $s$ a signal.
Further, let $e_1$ and $e_2$ be executions of $C$
  with the same input traces as $e$, but with
  $e_1$ additionally having transient faults within $R$ at $s$ up to some time $\tau_1 \in R$ and $e_2$
  having at least one transient at a time $\tau_2 \in R$ at $s$, where $\tau_1 \leq \tau_2 \leq t' - |C| \varepsilon - \gamma$.

Then for all value regions $R'$ of execution $e$ and all signals $s'$, if~$s'$ has value $\meta$ at some time within $R'$ in execution $e_1$, then it does so at some time within $R'$ in execution $e_2$.
\end{lem}
\begin{proof}
Within the same value region, both transient faults
  cause the same signals to become $\meta$, given that $\tau_1,\tau_2$ are sufficiently far from the value region's boundary $t'$ to allow for propagation of $\meta$ with delay $\varepsilon$ (at most $|C|\varepsilon$ time):
  this follows from the fact that the circuit's signal
  values and set of scheduled actions are identical at the start of the first transient in $e_1$ and in $e_2$.

Further, a signal with value $\meta$ remains so unless it
  is set again to a Boolean value by a production rule.
This can only happen by its guard becoming true right after a transient fault.
Since, $\tau_1 \leq \tau_2$, and both times are in the same
  value region, any event scheduled (and not canceled) after the transient fault at $\tau_2$ must also be scheduled (and not canceled) after the transient faults that occur until time $\tau_1$:
  signals have the same Boolean values and remain stable for a longer time in $e_1$ than in $e_2$.

The argument is inductively repeated for each subsequent value region of execution $e$.\qed
\end{proof}

We are now in the position to show the section's main result.

\begin{proof}[Proof of Theorem \ref{thm:main}]
Letting $\varepsilon \to 0$ and $\gamma \to 0$, we have from Lemma \ref{lem:tau1M-to-tau2M} that if a transient fault at a signal $s$ at a time $\tau_1 \in R$
  that causes $\meta$ at a signal $s'$ then a transient fault at a signal $s$ at a time $\tau_2 \in R$, where $\tau_1 \leq \tau_2$, also causes $s'$ to become $\meta$ at some point in time.
The theorem's statement then follows from the definition of a postfix of $R$.\qed
\end{proof}

\subsection{Automated computation of susceptible regions}\label{sec:tool}

Theorem \ref{thm:main} directly leads to an algorithm that marks all
  sensitivity windows, i.e., susceptible times, within an execution prefix:
  for each non-output signal $s$, and for each value region $R$, it finds
  per bisection of repeatedly inserting transient faults
  the boundary of non-susceptible times (on the left within $R$)
  and susceptible times (on the right within $R$).
The algorithm's time complexity is determined by one bisection per region (with precision $\varepsilon > 0$), i.e., $\sum_R \log \frac{|R|}{\varepsilon}$, as opposed to a naive search that injects a fault every $\varepsilon > 0$ with a complexity inversely proportional to $\varepsilon$.
Moreover, the naive algorithm may miss susceptibility windows smaller than $\varepsilon$, while our algorithm provably finds all such windows.

To test the algorithm's use on typical circuit instances we have implemented it in Python \cite{github:async-and-faults}:
  given a circuit, input traces, the set of monitored signals,
  as well as a time $T$ until which an execution
  prefix is to be generated, it outputs a figure with all
  susceptible windows highlighted in blue as well as the percentage of
  the length of the susceptible windows in the execution prefix
  (by default excluding the monitored signals, but with the possibility to include them).
This value corresponds to the probability of a transient fault causing
  an $\meta$ value at a monitored signal, i.e., \emph{the probability to fail} $P(\text{fail})$, given a uniform distribution
  of a single transient on all times in the execution prefix and
  on all signals that are not monitored signals (by default; alternatively on all signals).
Clearly, though, the uniform distribution can be easily
  replaced by a more involved distribution.
Towards this direction, the tool also outputs the probability per signal.
This allows one to compute a weighted average, e.g., depending on driver strength
  or shielding at certain signals.
Figure \ref{fig:trace-check} shows the tool's output for the previous example
  of the 3-stage linear pipeline with sensitivity windows marked in blue.

\begin{figure}[h]
    \centering
    \includegraphics[width=0.7\columnwidth]{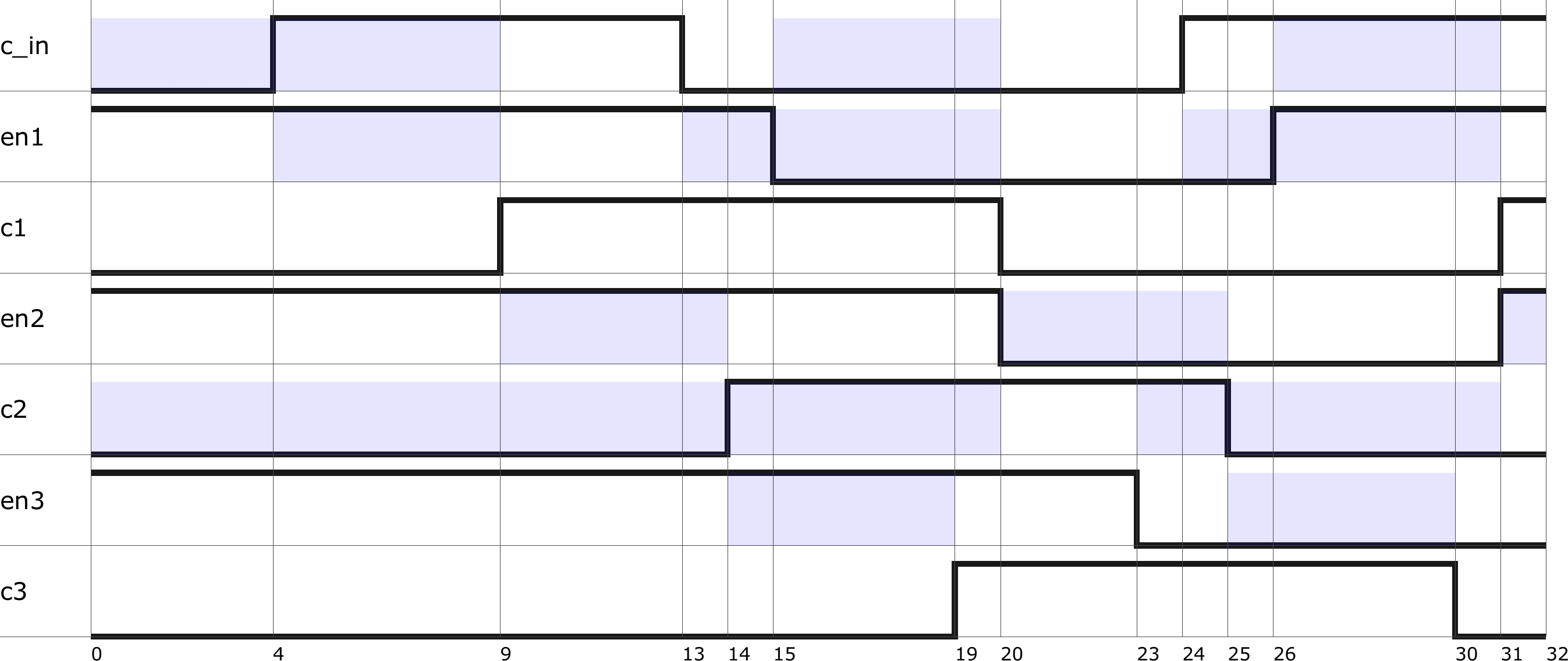}
    \caption{Execution prefix of linear 3-stage pipeline until time $T=32$ with sensitivity windows marked in blue. Monitored signals are {\tt c1} and {\tt c3}.
    Here, $P(\text{fail}) = 0.54375$.}
    \label{fig:trace-check}
\end{figure}

A fault occurring at any point of the blue sensitivity windows will drive one (or more) of the circuit's monitored signals to~$\meta$. A fault hitting any other region (excluding the monitored signals) will be masked and will not reach the monitored signals of the circuits.
Observe that in this example all sensitivity windows are trivial postfixes of regions: a region is either fully non-susceptible or susceptible; in general this does not necessarily hold.

\subsection{Comparison of fault-tolerance depending on speed}\label{sec:linear}

We next illustrate the use of our tool for a linear pipeline, where we vary the source and sink latencies.
Inverter delays are symmetric and normalized to~1 and Muller C-element latencies are set to~5 inverter delays.
The results are shown in Figure~\ref{fig:muller3_sweep_3dplot}, with cuts in Figures~\ref{fig:muller3-sweepSink} and~\ref{fig:muller3-sweepSource}. 
The length of the execution prefix has been chosen sufficiently high, to account for a sufficiently long time for $P(\text{fail})$ to be dominated by the periodic operation of the circuit rather than the initial transient phase: $T=500$ in the overview plot and $T=1000$ in the detailed sweeps.

\begin{figure}[h]
    \centering
    \includegraphics[width=0.3\columnwidth]{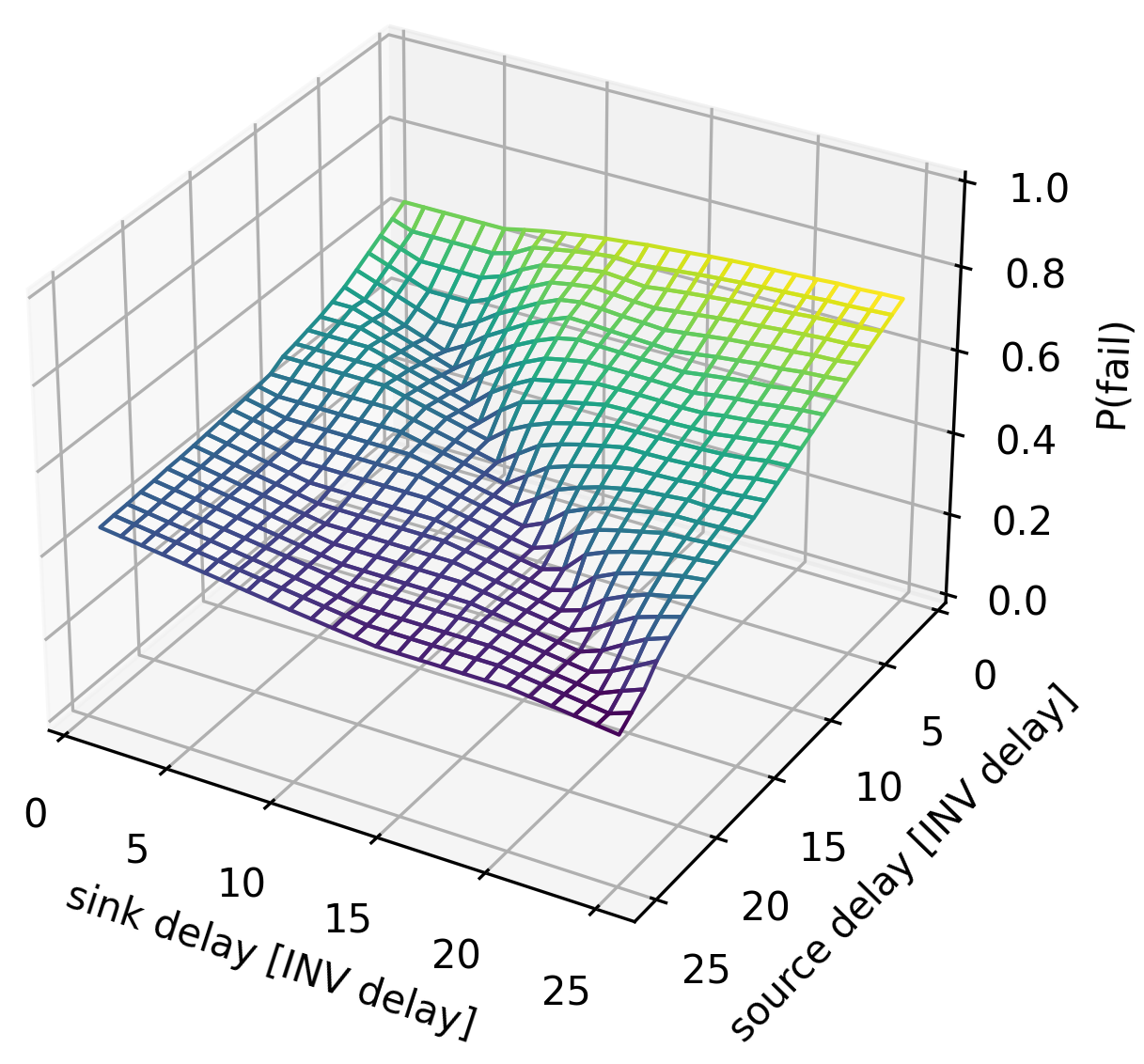}
    \vspace{-0.2cm}\caption{Influence of source and sink speed on $P(\text{fail})$. Linear 3-stage pipeline with delays as follows: 1 ({\tt INV}), 5 ({\tt MCE}), varying source and sink delays. $T = 500$.}
    \label{fig:muller3_sweep_3dplot}
\end{figure}

\begin{figure}
    \centering
    \begin{minipage}{0.45\textwidth}
        \centering
        \includegraphics[width=0.9\textwidth]{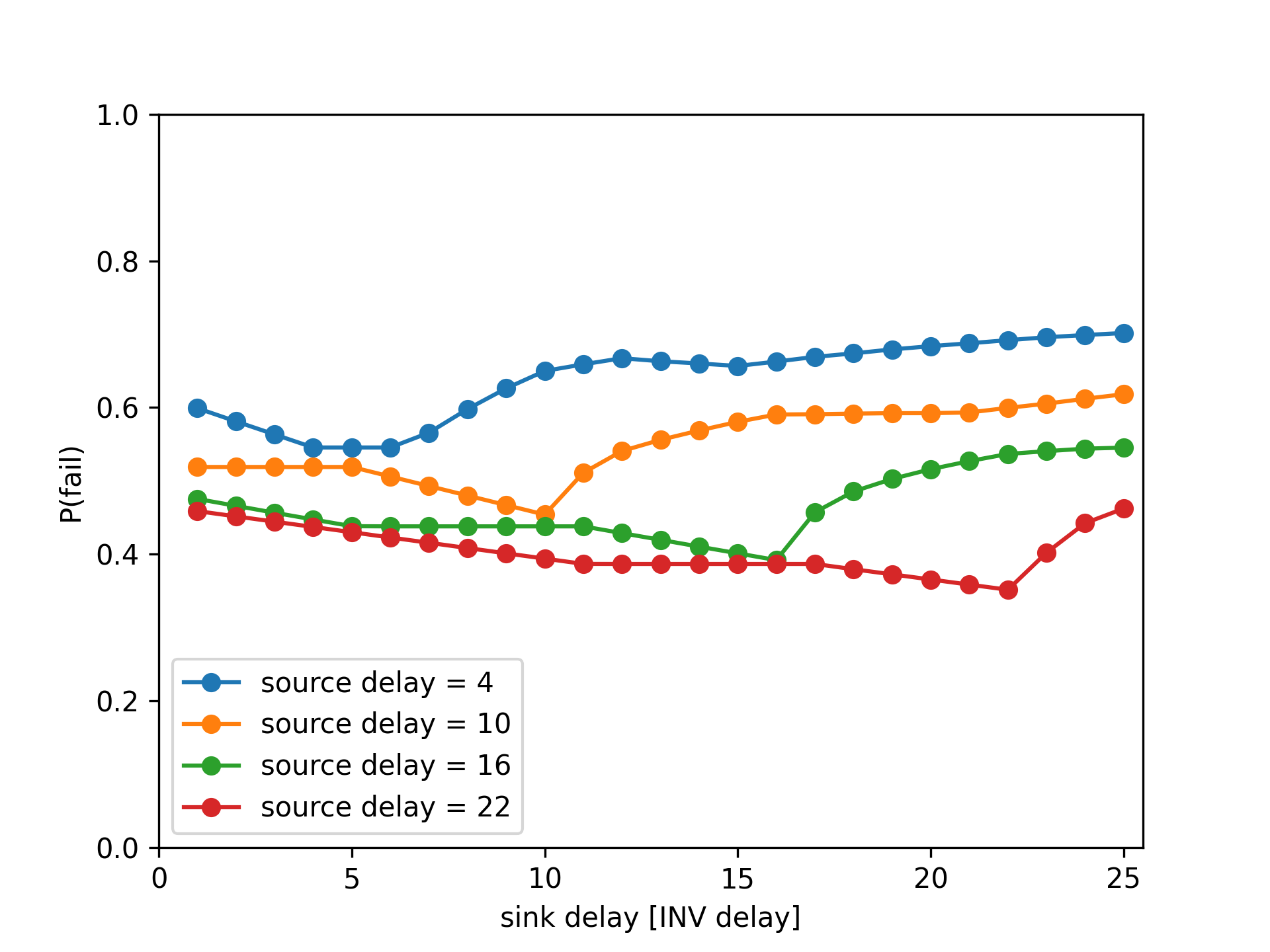}
        \vspace{-0.2cm}\caption{Influence of sink speed on $P(\text{fail})$. Linear 3-stage pipeline with delays: 1 ({\tt INV}), 5 ({\tt MCE}), 4 different source delays, varying sink delay. $T=1000$.}
        \label{fig:muller3-sweepSink}
    \end{minipage}\hfill
    \begin{minipage}{0.45\textwidth}
        \centering
        \includegraphics[width=0.9\textwidth]{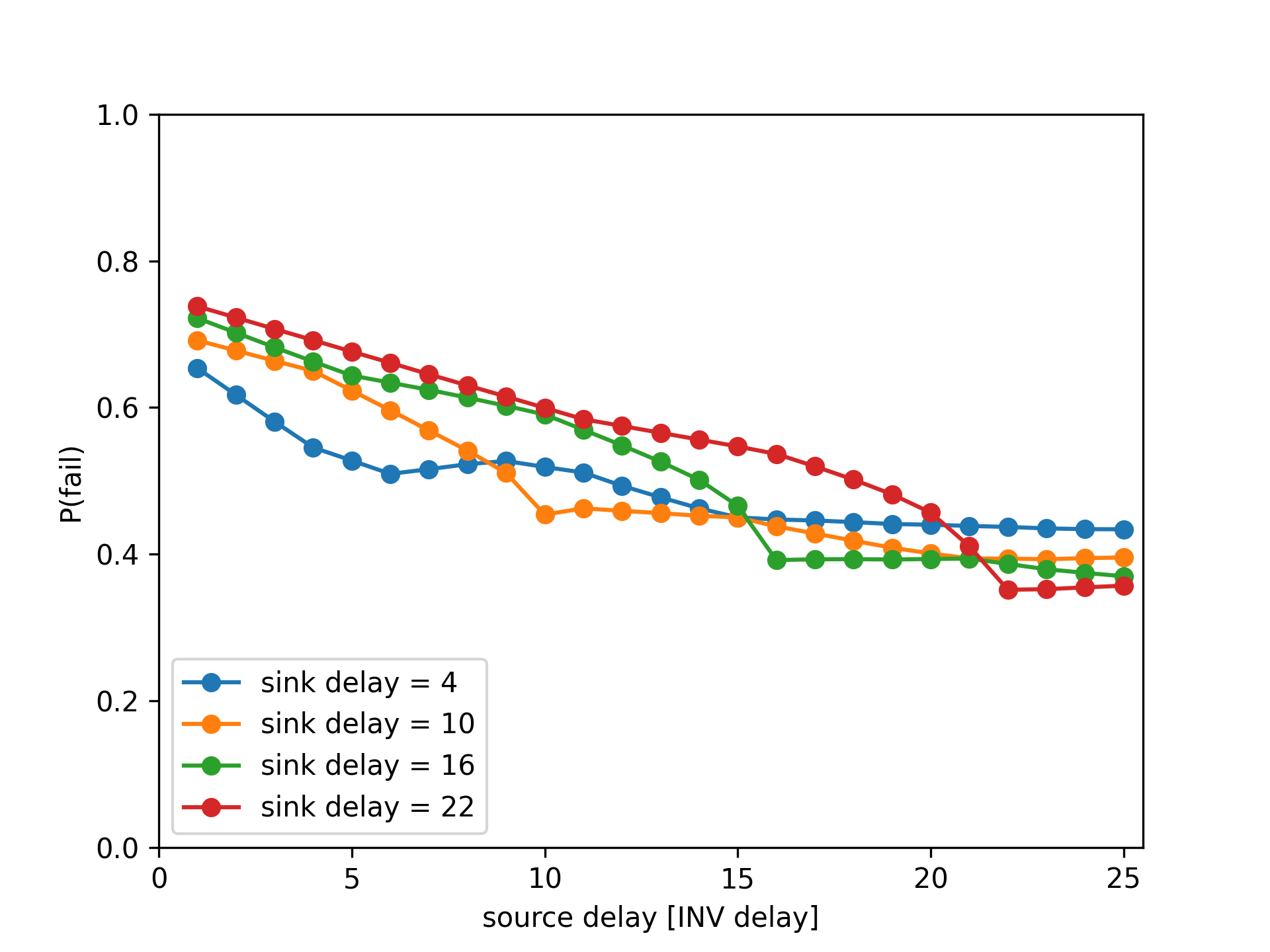} 
        \vspace{-0.2cm}\caption{Influence of source speed on $P(\text{fail})$. Linear 3-stage pipeline with delays: 1 ({\tt INV}), 5 ({\tt MCE}), 4 different sink delays, varying source delay. $T=1000$.}
        \label{fig:muller3-sweepSource}
    \end{minipage}
\end{figure}

Figure~\ref{fig:muller3_sweep_3dplot} shows an overview of the behavior of the circuit under a stable environment, be it fast or slow, and how the circuit reacts when there is an unbalance between the speeds of source and sink. The $z$-axis displays the probability of an~$\meta$ value presence at any of the monitored signals. The $x$ and $y$-axes represent the speeds (latencies) of sink and source, respectively, in time units. The pattern of the plot is best visualized when having the latter axis inverted.
The diagonal of the frame where both sink and source latencies are equal to 1 (fast) to where they are both 25 (slow) represents the stable/balanced environment, i.e., the source provides \emph{tokens} with the same speed as the sink provides the \emph{acknowledgment}. The figure indicates that $P(\text{fail})$ is high when the environment is stable and fast, and decreases as it gets stable and slow. When the environment is balanced, the MCEs in the circuit are not waiting for either the \emph{data} or the \emph{ack} signals; both are supplied within short intervals of time from each other. Since the waiting phases are those where the MCE operates in the vulnerable storage mode (inputs mismatching), one observes that reducing the waiting period decreases $P(\text{fail})$.

The environment imbalance is divided further into 2 modes of operation. On the right side (for relatively low source delay) of the figure, the circuit is operating in \emph{bubble-limited} mode, where the sink's response to the source's new tokens is slow. On the left half of the figure, the sink's activity is faster than the source's, driving it in \emph{token-limited} mode.

 The vulnerability of the bubble-limited mode can be seen more clearly in Figure~\ref{fig:muller3-sweepSource}; this is where the system is most prone to failure. The probability $P(\text{fail})$ varies from around 60-80\%, where it reaches the maximum when the sink delay is equal to 22 while source delay is 1 (maximum imbalance).
Similarly, the token-limited mode falls near the sink latency of~1 in Figure~\ref{fig:muller3-sweepSink}, varying from around 40-60\%.
The latter figures show several cross-sections of the 3D-plot from Figure~\ref{fig:muller3_sweep_3dplot}. In addition to mapping the token-limited and the bubble-limited areas to these 2 graphs, we can also spot the points belonging to the \emph{balanced environment} diagonal in the frame in Figure~\ref{fig:muller3_sweep_3dplot}.
These points are where the abrupt changes of behavior of each line occur, and consequently we can pinpoint where one region of the mode of operation ends and the other starts.

Finally, Figure \ref{fig:p_per_sig} shows the fault probabilities per signal of the
  linear pipeline as reported by our tool for varying source and sink delays (fast, normal, and slow). It allows us to give a more detailed interpretation of our observations in Figure~\ref{fig:muller3-sweepSink}. 
The probabilities for the monitored signals {\tt c1} and {\tt c3} are always~$1.0$ as
  by definition of the fault probabilities.

\begin{figure}[t]
    \centering
    \includegraphics[width=0.25\columnwidth]{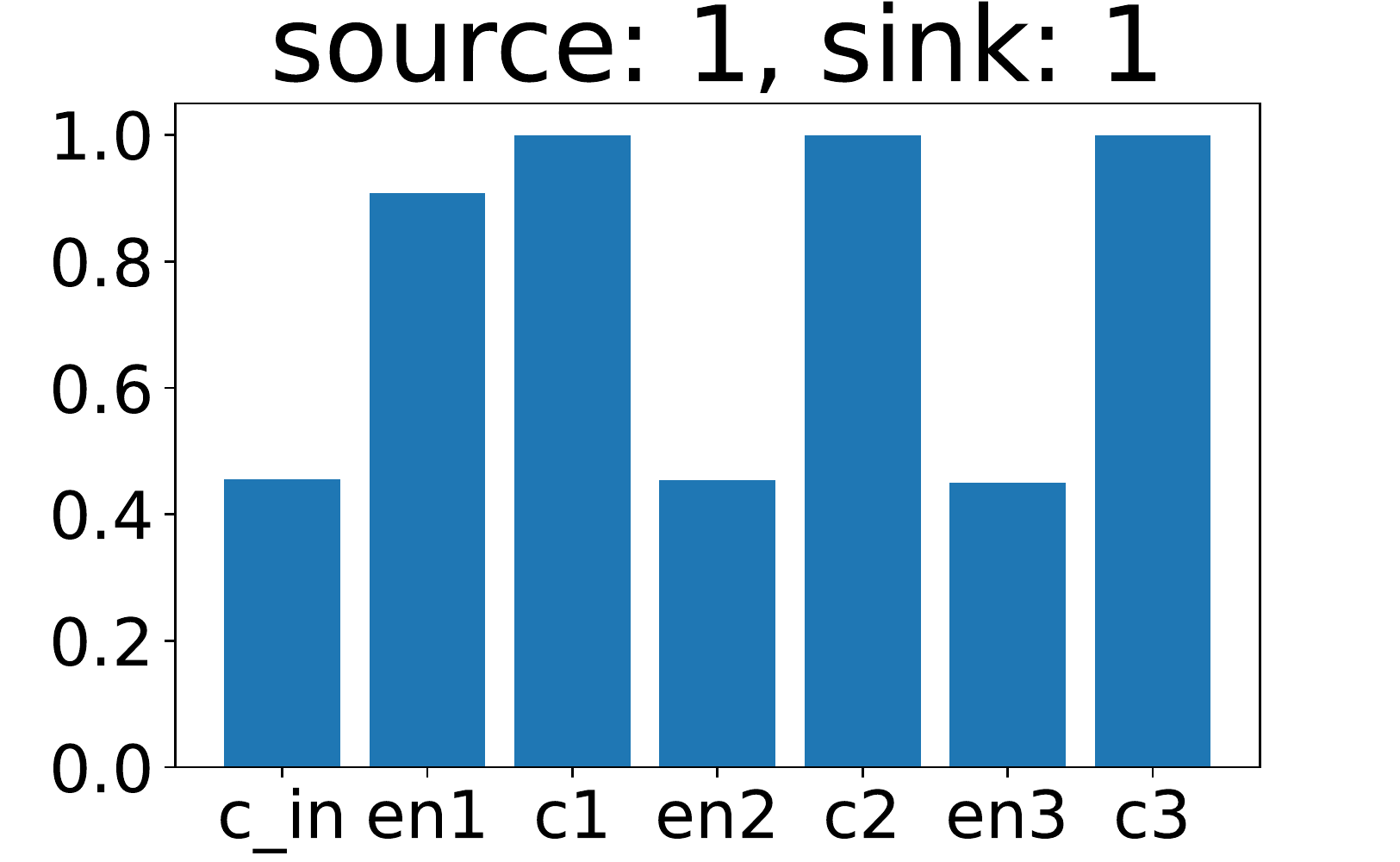}
    \includegraphics[width=0.25\columnwidth]{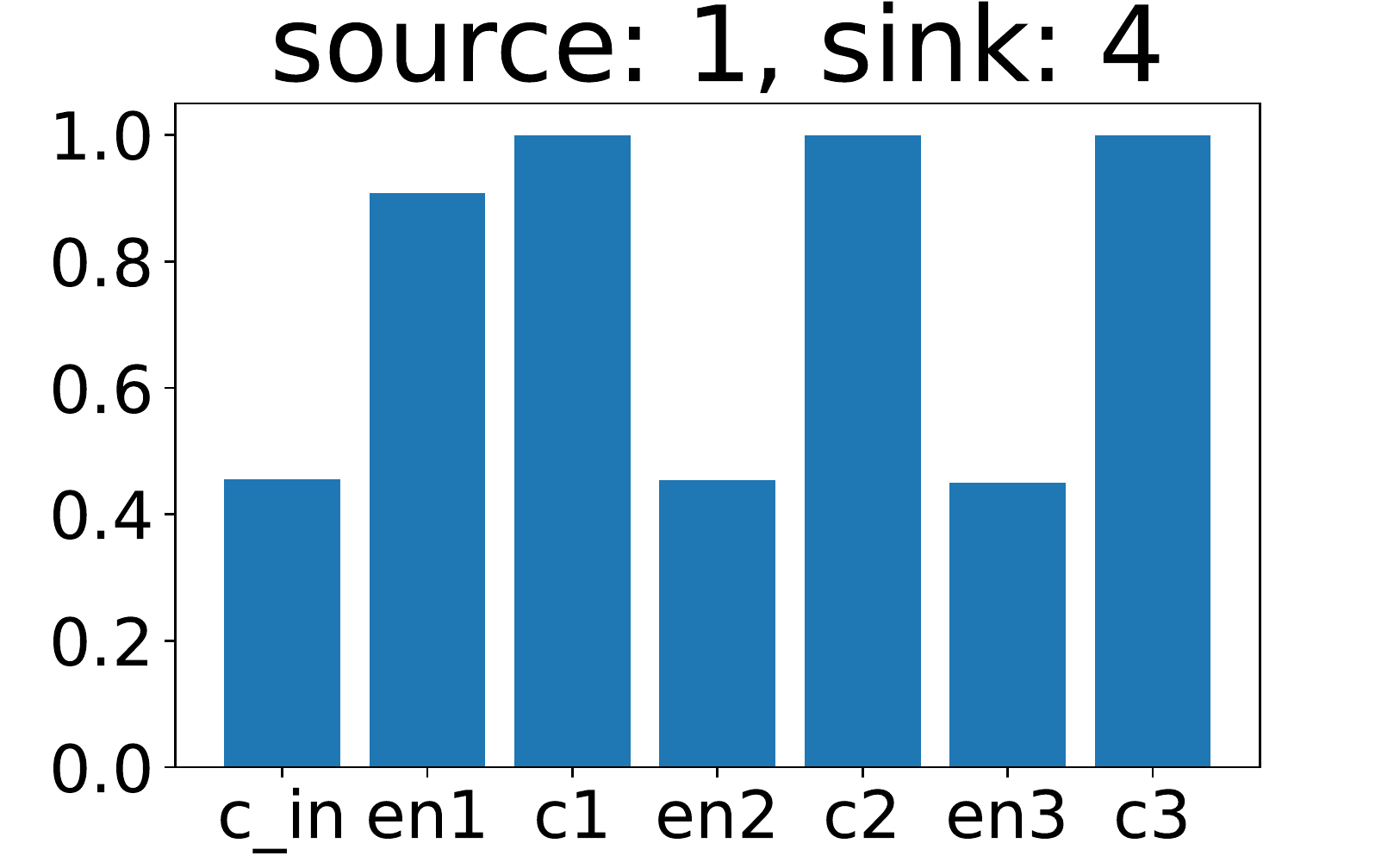}
    \includegraphics[width=0.25\columnwidth]{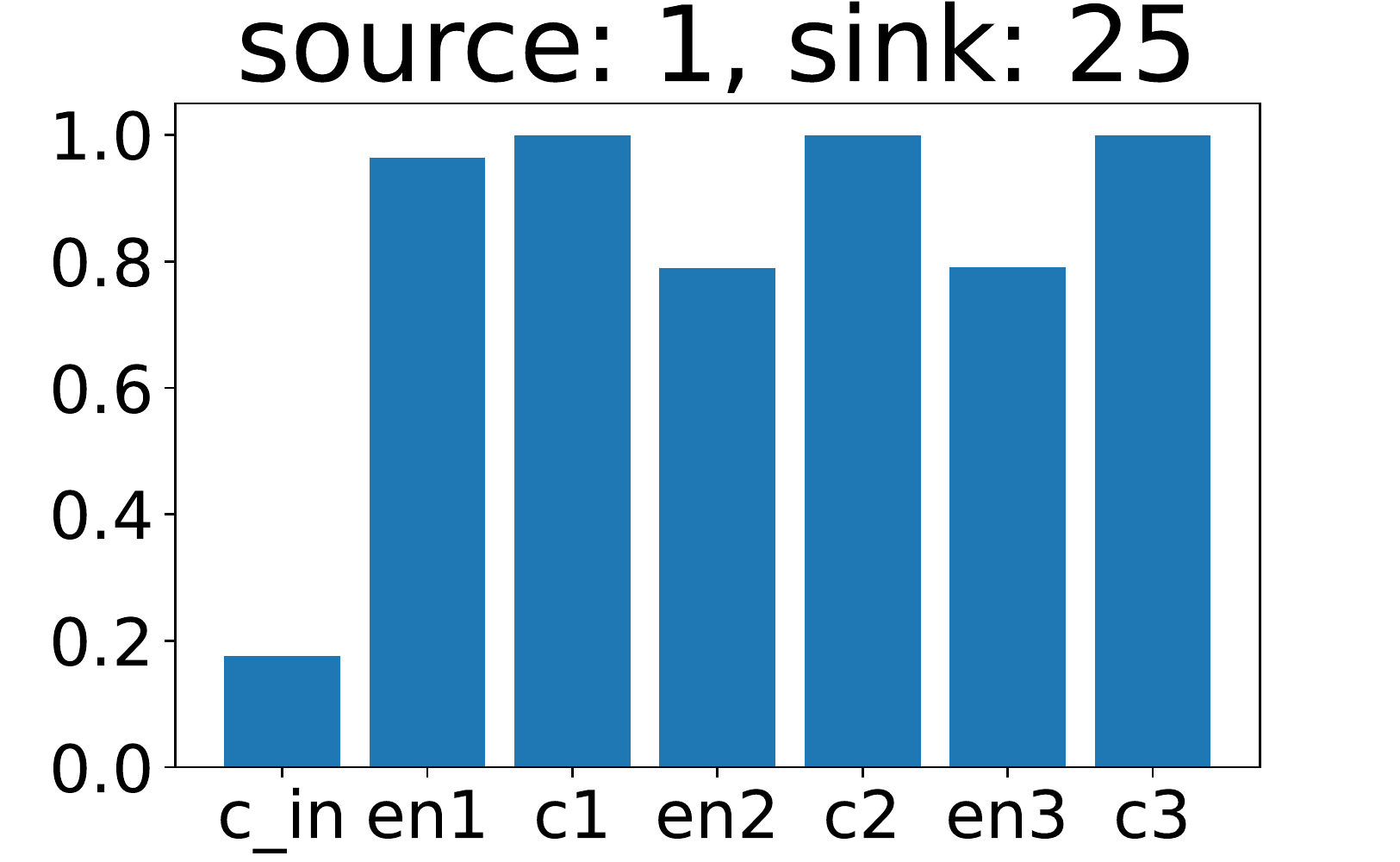}\\
    \vspace{0.2cm}
    \includegraphics[width=0.25\columnwidth]{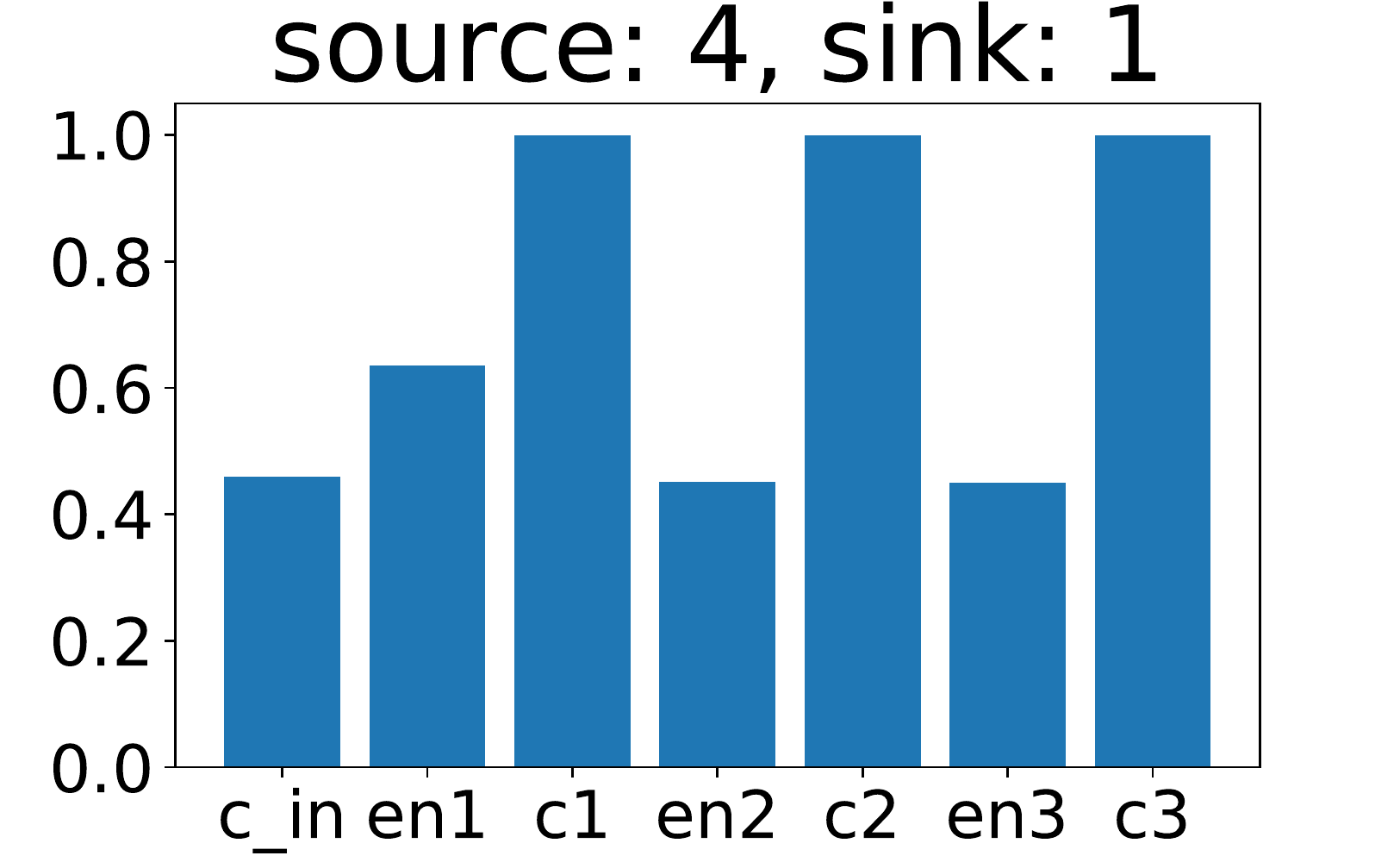}
    \includegraphics[width=0.25\columnwidth]{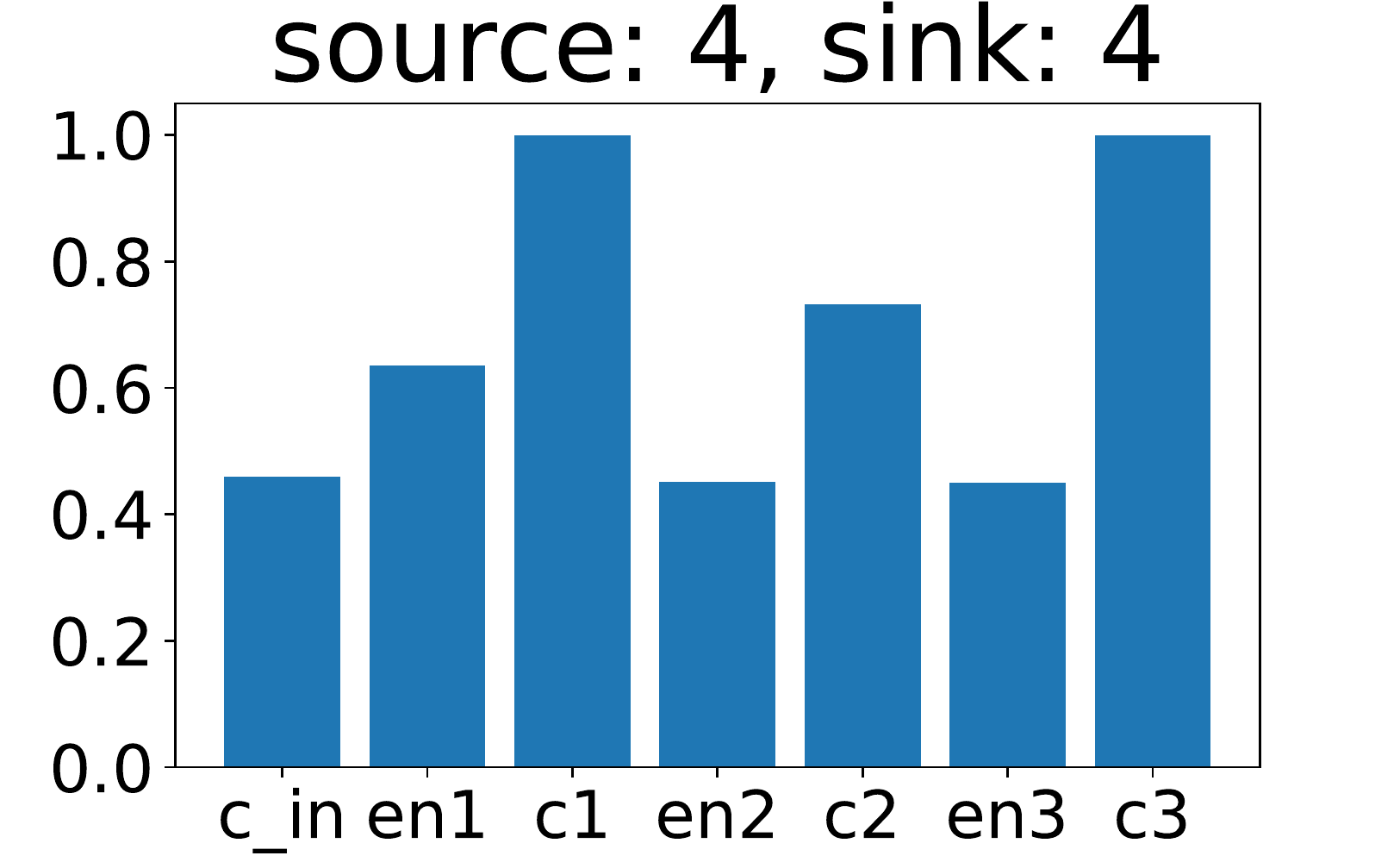}
    \includegraphics[width=0.25\columnwidth]{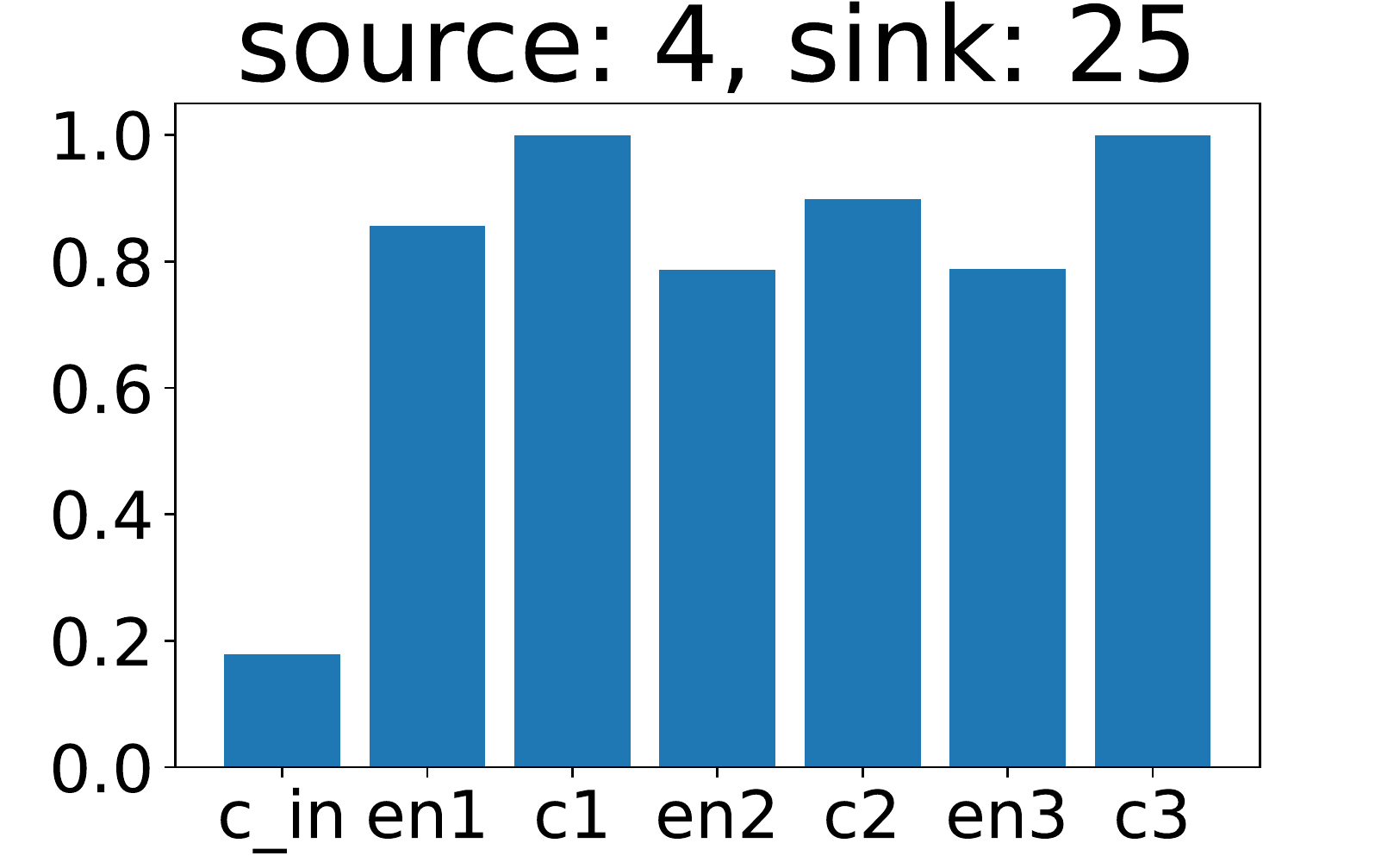}\\
    \vspace{0.2cm}
    \includegraphics[width=0.25\columnwidth]{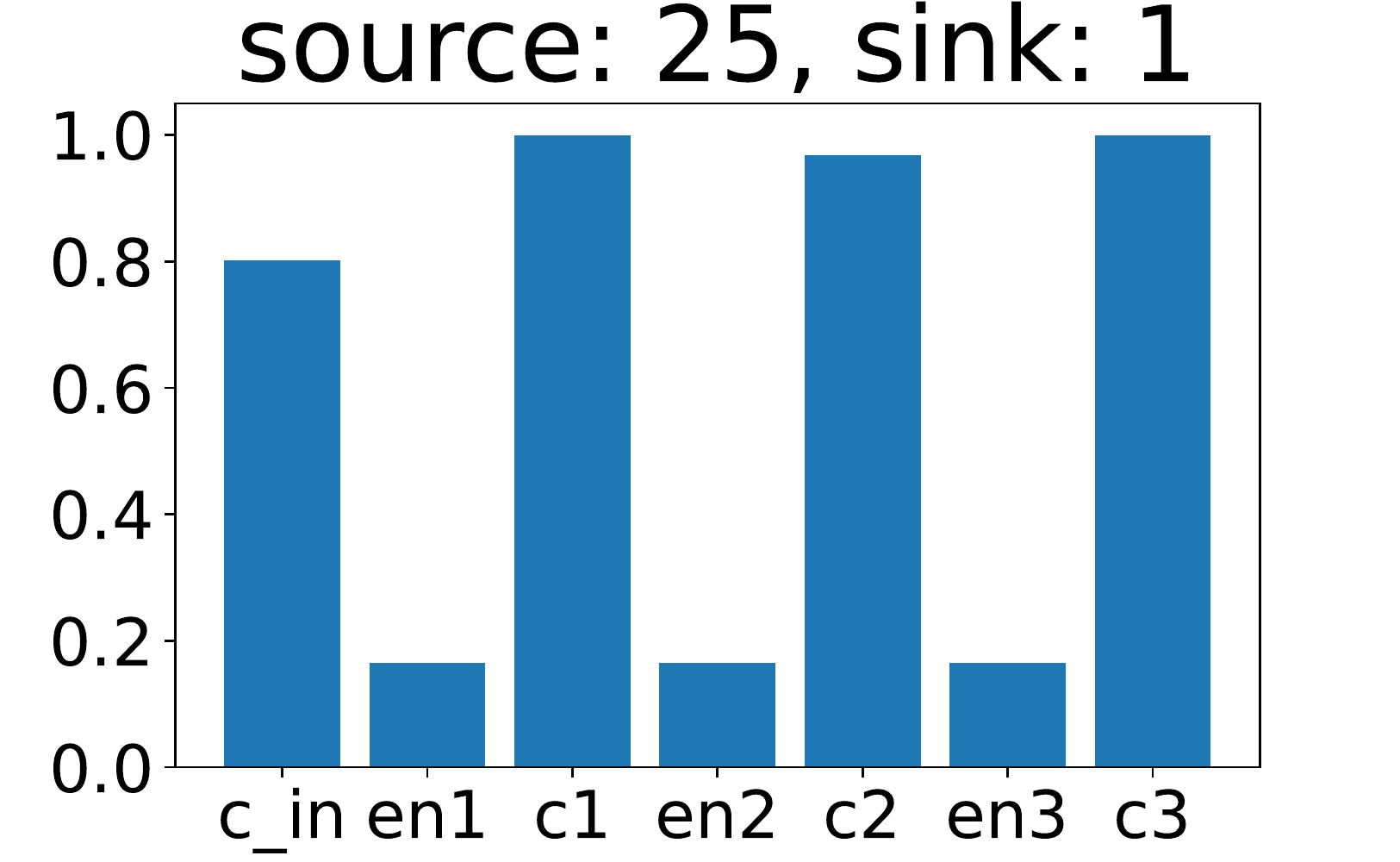}
    \includegraphics[width=0.25\columnwidth]{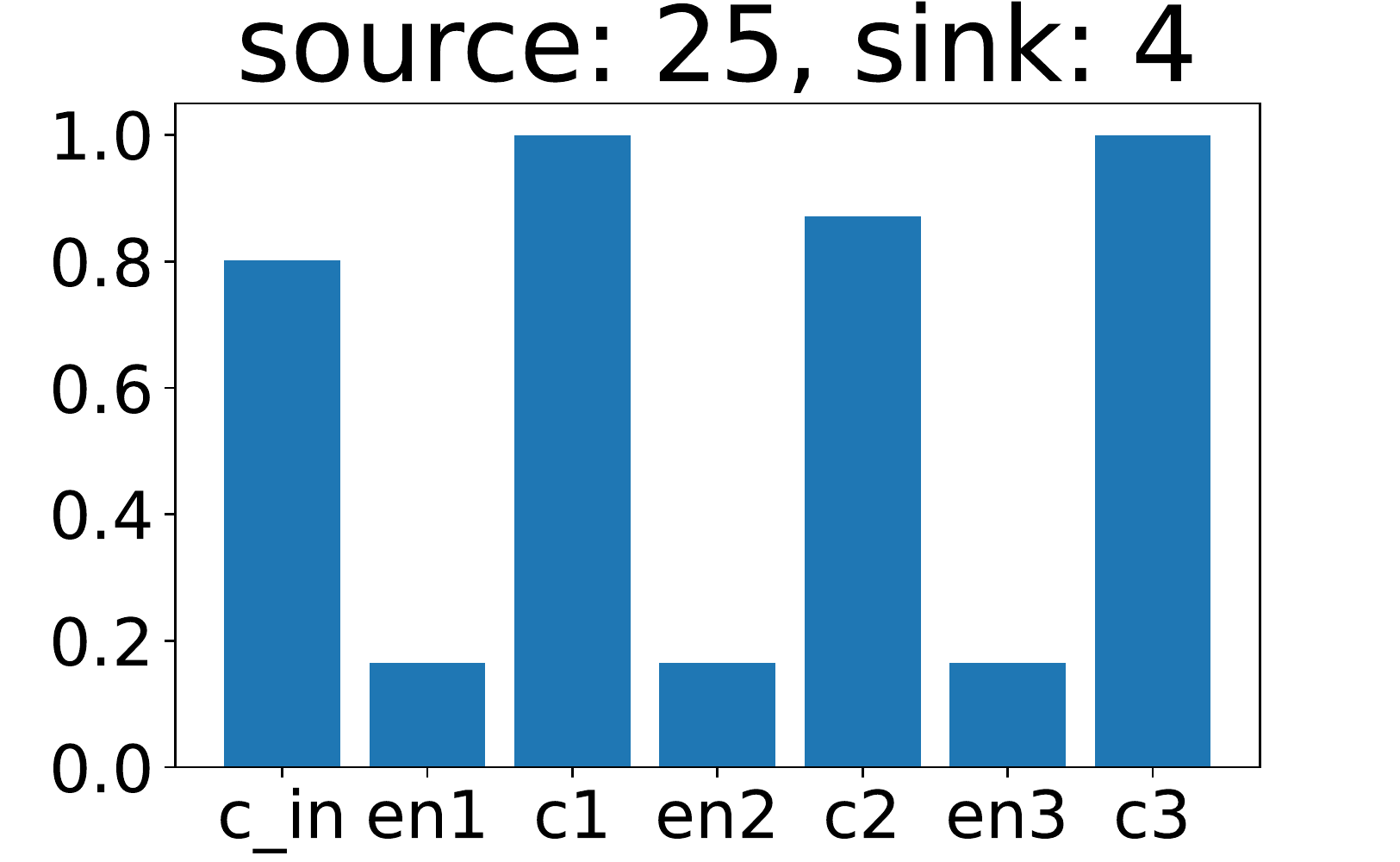}
    \includegraphics[width=0.25\columnwidth]{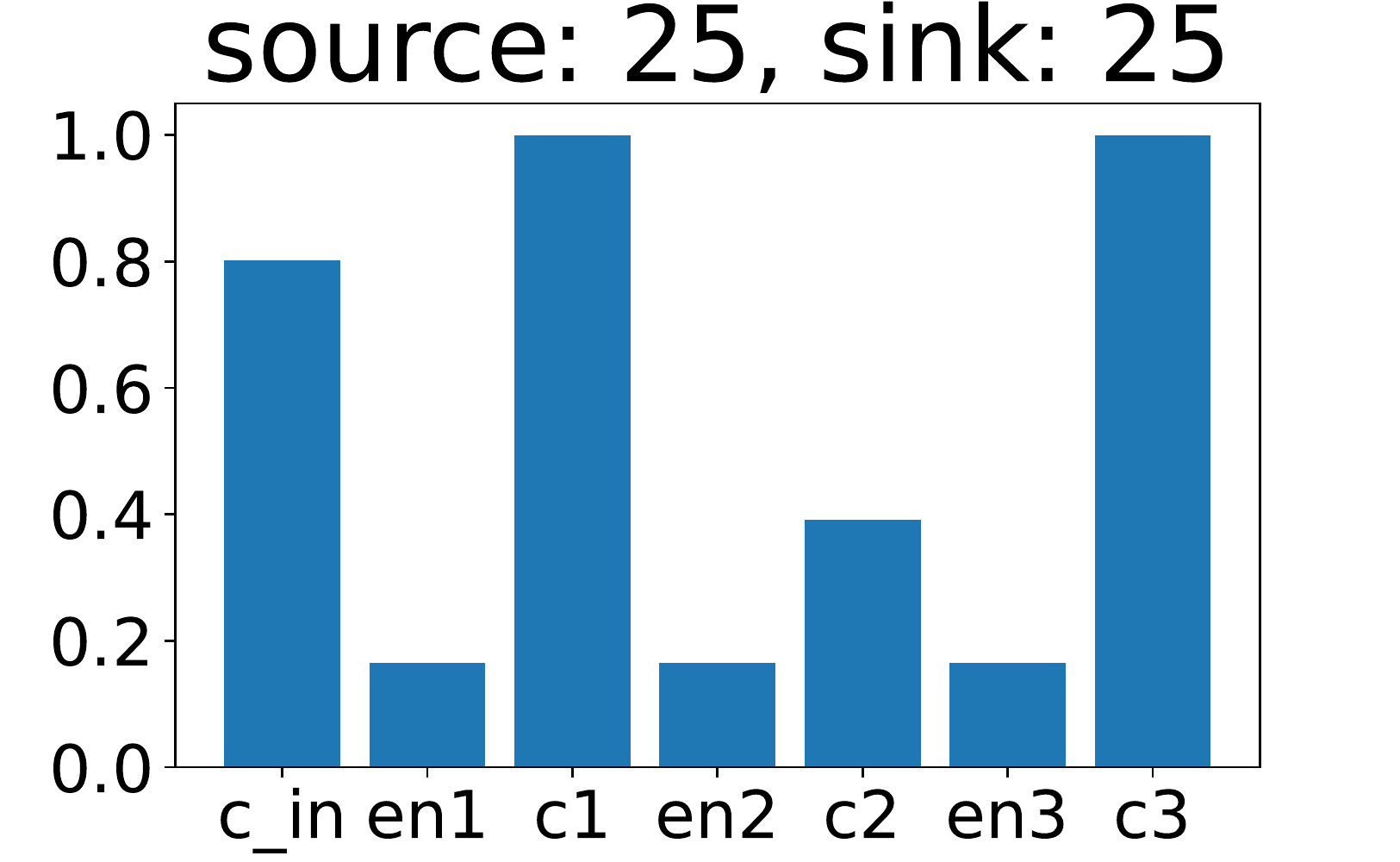}
    \vspace{-0.1cm}
    \caption{Circuit fault probability per signal of the linear 3-stage pipeline with varying source and sink delays. Delays: 1 ({\tt INV}), 5 ({\tt MCE}), and 3 different source and sink delays indicated in the figure. $T=1000$.}\vspace{-0.5cm}
    \label{fig:p_per_sig}
\end{figure}

Interestingly, {\tt c2} has a high fault probability, too. For fast sink this can be explained as follows: {\tt MCE3} spends most of the time in the vulnerable storage mode, waiting for a transition on {\tt c2}. As soon as one occurs, it triggers a transition on {\tt c3} which, after the short sink delay, puts {\tt MCE3} back to storage mode. This only leaves a very short time window where {\tt MCE3} is masking faults at {\tt c2}. The enable signals, in turn, are only vulnerable during those short windows, thus showing low fault probability, especially when the source delay is high, see subplot ``source:25, sink:1''.

Recall that we have chosen a relatively large switching delay for the MCE, and our pessimistic model assumes the MCE to be vulnerable during the whole switching duration. This explains why in general $P(\text{fail})$ increases for faster operation speed: the proportion of the sensitive MCE switching phase increases. This can be most directly observed for the balanced cases.

For the other imbalanced extreme with ``source:1, sink:25'' we observe high fault probability for the enable signals. This is not surprising, since now the MCEs spend most of the time waiting for transitions on these signals. A fault probability of $1.0$ for {\tt c1} and {\tt c3} is also unsurprising, due to our definitions, as mentioned already. Quite unexpected, on the first glance, is the fact that {\tt c2} again shows high fault probability, even though we can assume good masking by {\tt MCE3} for that input. The reason here is that, via {\tt INV2}, faults from {\tt c2} directly propagate to {\tt en1} which is known to have low protection by masking. As a result, we see a generally high fault probability in this mode.

For additional simulations and larger circuits, we refer the reader to the appendix: a cyclic version of the Muller pipeline is discussed in Section~\ref{sec:ring} and higher-bit-width modules are discussed in Section~\ref{sec:multi}.

\section{Conclusion}
\label{sec:conclusion}

We have shown that the regular operation of circuits can be decomposed into time windows within which faults are equivalent in that their effect (as perceived at some selected monitoring signals) remains the same. These time windows are bounded by an arbitrary bound on the left and a regular signal transition on the right. 
Consequently, for determining the effect of transient faults on a circuit, a single bisection between each pair of neighboring signal transitions is sufficient to determine all sensitivity windows.

The approach has two advantages over standard sweeping approaches to find sensitive regions: (i) it provably finds all sensitivity windows, no matter how small they are. Sweeping by contrast always leaves the possibility open that a small window may exist between two samples.
(ii) It outperforms sweeping in that a small grid of samples is not necessary: many (large) windows require only a single sample via our method, and at most a bisection.

Based on this result we have developed a Python-based tool that, starting from a production-rule based circuit description, systematically explores its resilient and its vulnerable windows (along with the respective fault effects). 
The relative size of the windows is then used to predict the proportion of (random) faults that will be effective, and thus, given a fault rate, the failure rate.
Since our approach allows identifying the windows individually, it is possible to attach weights to the diverse nodes to account for different susceptibility (drive strength, e.g.) in the overall prediction.
We have illustrated the function of our tool on a typical QDI circuit example which showed that the tool is efficient and allows for fast analysis.

While we focused on a proof-of-concept with smaller circuits, a next step is to run our method on larger circuits.
Another extension of our approach is to determine the constituent parameters for the window sizes. Since we determine all windows individually in our automated process, backtracking to the origins of the relevant signal transitions is possible. With that information we can determine in detail how individual parameters like circuit delays or pipeline load influence resilience and hence elaborate targeted optimizations.
Finally, work on improving the performance of the implementation is planned: the proposed algorithm is easily parallelizable since windows can be determined independently and hence concurrently.


\bibliographystyle{IEEEtran}
\bibliography{main}

\begin{thebibliography}{10}
\providecommand{\url}[1]{#1}
\csname url@samestyle\endcsname
\providecommand{\newblock}{\relax}
\providecommand{\bibinfo}[2]{#2}
\providecommand{\BIBentrySTDinterwordspacing}{\spaceskip=0pt\relax}
\providecommand{\BIBentryALTinterwordstretchfactor}{4}
\providecommand{\BIBentryALTinterwordspacing}{\spaceskip=\fontdimen2\font plus
\BIBentryALTinterwordstretchfactor\fontdimen3\font minus
  \fontdimen4\font\relax}
\providecommand{\BIBforeignlanguage}[2]{{%
\expandafter\ifx\csname l@#1\endcsname\relax
\typeout{** WARNING: IEEEtran.bst: No hyphenation pattern has been}%
\typeout{** loaded for the language `#1'. Using the pattern for}%
\typeout{** the default language instead.}%
\else
\language=\csname l@#1\endcsname
\fi
#2}}
\providecommand{\BIBdecl}{\relax}
\BIBdecl

\bibitem{martin1986compiling}
A.~J. Martin, ``Compiling communicating processes into delay-insensitive vlsi
  circuits,'' \emph{Distributed computing}, vol.~1, no.~4, pp. 226--234, 1986.

\bibitem{bainbridge2009glitch}
W.~J. Bainbridge and S.~J. Salisbury, ``Glitch sensitivity and defense of quasi
  delay-insensitive network-on-chip links,'' in \emph{2009 15th IEEE Symposium
  on Asynchronous Circuits and Systems}.\hskip 1em plus 0.5em minus 0.4em\relax
  IEEE, 2009, pp. 35--44.

\bibitem{lafrieda2004fault}
C.~LaFrieda and R.~Manohar, ``Fault detection and isolation techniques for
  quasi delay-insensitive circuits,'' in \emph{International Conference on
  Dependable Systems and Networks, 2004}.\hskip 1em plus 0.5em minus
  0.4em\relax IEEE, 2004, pp. 41--50.

\bibitem{peng2005efficient}
S.~Peng and R.~Manohar, ``Efficient failure detection in pipelined asynchronous
  circuits,'' in \emph{20th IEEE International Symposium on Defect and Fault
  Tolerance in VLSI Systems (DFT'05)}.\hskip 1em plus 0.5em minus 0.4em\relax
  IEEE, 2005, pp. 484--493.

\bibitem{monnet2007formal}
Y.~Monnet, M.~Renaudin, and R.~Leveugle, ``Formal analysis of quasi delay
  insensitive circuits behavior in the presence of seus,'' in \emph{13th IEEE
  International On-Line Testing Symposium (IOLTS 2007)}.\hskip 1em plus 0.5em
  minus 0.4em\relax IEEE, 2007, pp. 113--120.

\bibitem{monnet2005asynchronous}
------, ``Asynchronous circuits transient faults sensitivity evaluation,'' in
  \emph{Proceedings of the 42nd annual design automation conference}, 2005, pp.
  863--868.

\bibitem{huemer2020QDIwindows}
F.~Huemer, R.~Najvirt, and A.~Steininger, ``{Identification and Confinement of
  Fault Sensitivity Windows in QDI Logic},'' in \emph{Proceedings. Austrochip
  Workshop on Microelectronics 2020}.\hskip 1em plus 0.5em minus 0.4em\relax
  IEEE, 2020.

\bibitem{behal2021towards}
P.~Behal, F.~Huemer, R.~Najvirt, A.~Steininger, and Z.~Tabassam, ``Towards
  explaining the fault sensitivity of different qdi pipeline styles,'' in
  \emph{2021 27th IEEE International Symposium on Asynchronous Circuits and
  Systems (ASYNC)}.\hskip 1em plus 0.5em minus 0.4em\relax IEEE, 2021, pp.
  25--33.

\bibitem{tabassam2022set}
Z.~Tabassam and A.~Steininger, ``Set hardened derivatives of qdi buffer
  template,'' in \emph{2022 IEEE International Symposium on Defect and Fault
  Tolerance in VLSI and Nanotechnology Systems (DFT)}.\hskip 1em plus 0.5em
  minus 0.4em\relax IEEE, 2022, pp. 1--6.

\bibitem{jang2005seu}
W.~Jang and A.~J. Martin, ``Seu-tolerant qdi circuits [quasi delay-insensitive
  asynchronous circuits],'' in \emph{11th IEEE International Symposium on
  Asynchronous Circuits and Systems}.\hskip 1em plus 0.5em minus 0.4em\relax
  IEEE, 2005, pp. 156--165.

\bibitem{katelman2012rewriting}
M.~Katelman, S.~Keller, and J.~Meseguer, ``Rewriting semantics of production
  rule sets,'' \emph{The Journal of Logic and Algebraic Programming}, vol.~81,
  no. 7-8, pp. 929--956, 2012.

\bibitem{brzozowski2001algebras}
J.~A. Brzozowski, Z.~{\'E}sik, and Y.~Iland, ``Algebras for hazard detection,''
  in \emph{Proceedings 31st IEEE International Symposium on Multiple-Valued
  Logic}.\hskip 1em plus 0.5em minus 0.4em\relax IEEE, 2001, pp. 3--12.

\bibitem{github:async-and-faults}
R.~E. Shehaby, M.~F\"ugger, and A.~Steininger, ``Sensitivity analyzer for
  asynchronous logic ({SeAL}),'' \url{https://github.com/mfuegger/SeAL}, 2023.

\bibitem{martin1990limitations}
A.~J. Martin, ``The limitations to delay-insensitivity in asynchronous
  circuits,'' in \emph{Beauty is our business}.\hskip 1em plus 0.5em minus
  0.4em\relax Springer, 1990, pp. 302--311.

\bibitem{Manohar2017theorem}
R.~Manohar and Y.~Moses, ``The eventual c-element theorem for delay-insensitive
  asynchronous circuits,'' in \emph{2017 23rd IEEE International Symposium on
  Asynchronous Circuits and Systems (ASYNC)}, 2017, pp. 102--109.

\bibitem{beerel2010designer}
P.~A. Beerel, R.~O. Ozdag, and M.~Ferretti, \emph{A designer's guide to
  asynchronous VLSI}.\hskip 1em plus 0.5em minus 0.4em\relax Cambridge
  University Press, 2010.

\bibitem{gill2008performance}
G.~Gill, V.~Gupta, and M.~Singh, ``Performance estimation and slack matching
  for pipelined asynchronous architectures with choice,'' in \emph{2008
  IEEE/ACM International Conference on Computer-Aided Design}.\hskip 1em plus
  0.5em minus 0.4em\relax IEEE, 2008, pp. 449--456.

\bibitem{williams1987self}
T.~E. Williams, M.~Horowitz, R.~Alverson, and T.~Yang, ``A self-timed chip for
  division,'' in \emph{Stanford Conference on Advanced Research in VLSI}, 1987,
  pp. 75--96.

\end{thebibliography}

\clearpage
\appendix
\section{Appendix}

Section~\ref{sec:background} gives a brief overview on the main concepts of asynchronous circuits used in this work.
Sections~\ref{sec:ring} and~\ref{sec:multi} demonstrate the applicability of our approach to circuits different than the linear pipeline.

\subsection{Asynchronous circuits: overview on concepts used in this work}\label{sec:background}

\paragraph{QDI circuit}
Asynchronous circuits, unlike their synchronous counterparts, are not governed by a rigid time grid that centrally determines the communication between any two entities. Alternatively, data transfer follows a closed-loop control provided by a local handshake process that specifies when new data is available for sending and recognizes when old data has been processed. \emph{Delay-insensitive (DI)} circuits offer the ultimate timing flexibility by automatically adapting to gate and wire delays. Very few circuits, however, can actually be designed following the DI delay model, since such circuits can provably be constructed only from inverters (1 input) and Muller C-elements (MCEs, 2 inputs) when restricted to single-output gates \cite{martin1990limitations,Manohar2017theorem}. By adding the \emph{isochronic fork} constraint, one obtains the class of \emph{quasi delay-insensitive (QDI)} circuits, which is only lightly-constrained with respect to timing and provides sufficient expressivity. This timing assumption requires the delays of the individual paths of an isochronic fork to be about equal, assuring that a signal arrives at all ends of the fork at about the same time~\cite{beerel2010designer,martin1986compiling}.

\paragraph{Muller C-element}
The MCE is a fundamental building block in QDI circuits, as it can be seen as the simplest form of storage element, and it is crucially used in the so called \emph{completion detection (CD)} units. It can also serve as a control unit in some QDI buffer templates. The MCE can be considered as an AND gate with hysteresis. In case of matching inputs, it sets the output to this corresponding value; when inputs don't match, the output retains its previous logic level.

\paragraph{QDI pipeline}
In a QDI pipeline, data is issued by a source and travels through sequential stages of latches, which are sometimes separated by logic function units, until it reaches the sink (in absence of logic functions, the pipeline acts as a data queue). Figure~\ref{fig:stage} shows the components of one such stage. When new data is detected at the input side, the buffer captures it only when an acknowledgement from the next stage is received, indicating that old data has been stored. The CD then signals the receipt of this new data item by issuing an acknowledgment signal to its previous stage, closing the control loop.

Since there is no explicit request signal in the QDI design style, data must be encoded using a so-called DI code. The encoding depends on the communication protocol in use. The \emph{return-to-zero (RTZ)}, or \emph{4-phase}, protocol is commonly used along with the \emph{dual-rail} encoding scheme, where one data bit is encoded on two wires, namely the \emph{true} and \emph{false} rails: Any two data items are separated by a \emph{spacer}. This spacer is encoded by logic ‘0’ on all rails and carries no information. Only one rail for each bit can be set to logic ‘1’ at a time; having both rails set to ‘1’ violates the protocol and is considered \emph{illegal}. The CD indicates validity and completeness of a data item or a spacer, issuing the appropriate \emph{ack} to the other stage. The CD is made up of a simple OR gate in the case of a single bit, and employs additional logic for higher bit-widths, as shown in Figure~\ref{fig:CD}.

\begin{figure}
    \centering
    \begin{minipage}{0.45\textwidth}
        \centering
        \includegraphics[width=\textwidth]{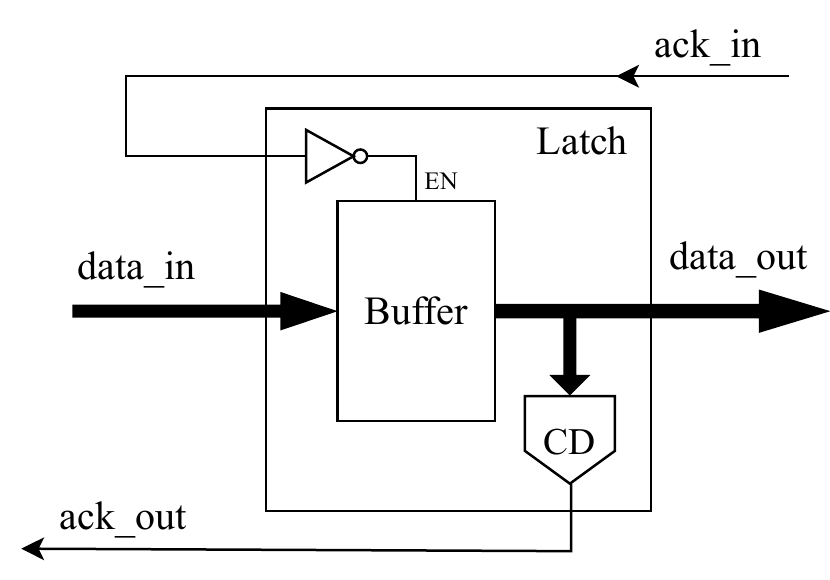}
        \caption{QDI pipeline stage with latch components}
        \label{fig:stage}
    \end{minipage}\hfill
    \begin{minipage}{0.45\textwidth}
        \centering
        \includegraphics[width=\textwidth]{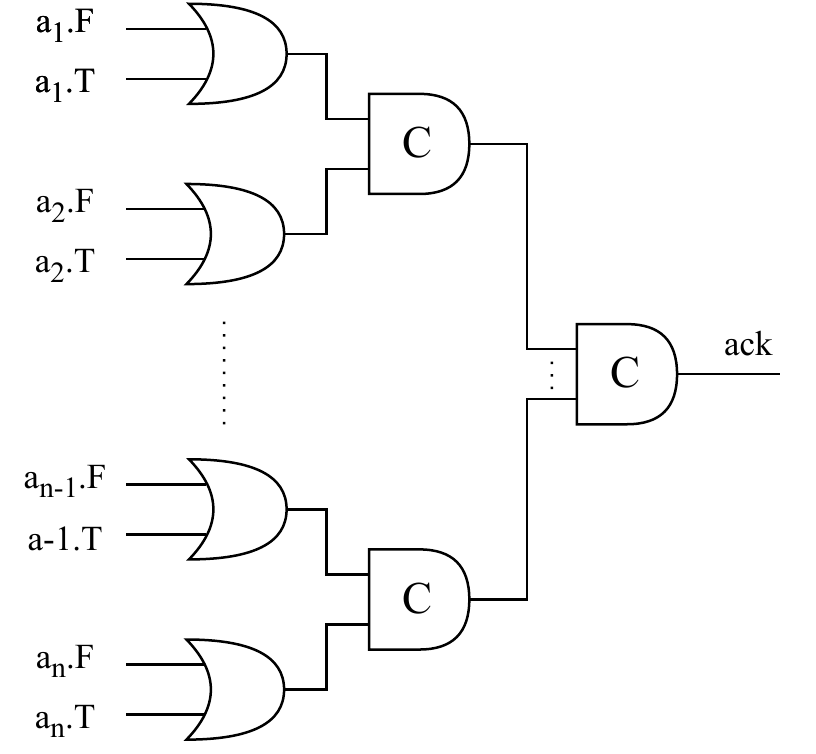} 
        \caption{The internal structure of a multi-bit CD using OR gates and MCEs}
        \label{fig:CD}
    \end{minipage}
\end{figure}

\subsection{Comparison of fault-tolerance for Muller pipeline rings}\label{sec:ring}

Another common asynchronous pipeline
construct is rings. We interpret the pipeline operation to implement a 4-phase QDI protocol in the following. Shown in Figure~\ref{fig:muller3_ring}, is a 3-stage Muller pipeline where one \emph{data token}, one \emph{spacer}, and one \emph{bubble} keep rotating. Note that when using the term \emph{token} on its own, it encompasses a data token along with a spacer, so 1 token means one of each.
It is possible to also interpret a token to follow the 2-phase communication protocol, and in this case we would double this count.

\begin{figure}[h]
    \centering
    \includegraphics[width=0.6\columnwidth]{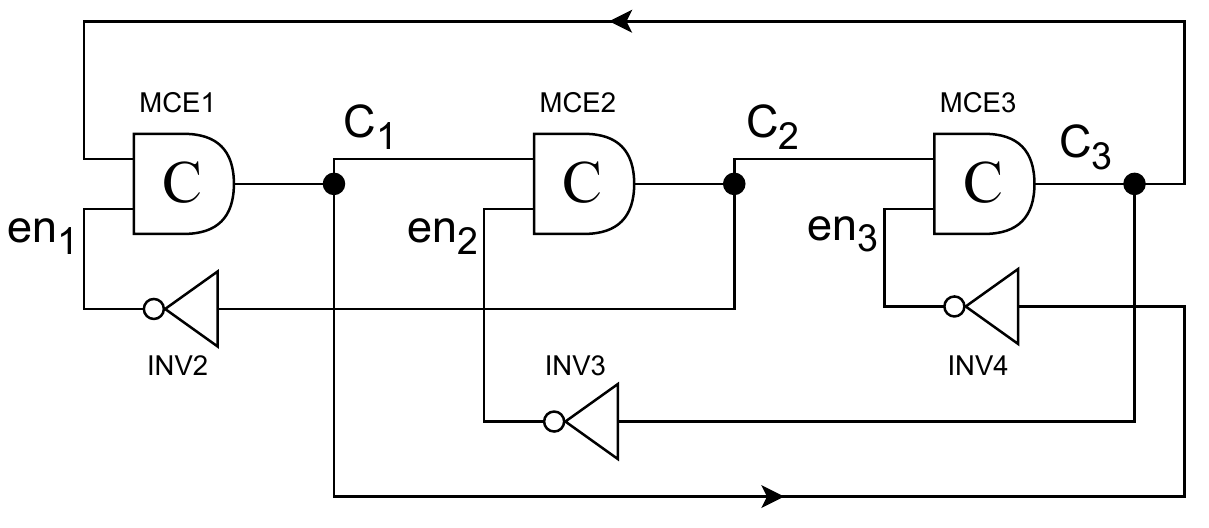}
    \caption{Ring 3-stage pipeline. The delays are set to 1 ({\tt INV}) and 5 ({\tt MCE}).}
    \label{fig:muller3_ring}
\end{figure}

In order to study the resilience of a Muller ring w.r.t. the activity inside the pipeline, 
we need to vary its \emph{occupancy}, i.e., the number of data items revolving in it \cite{gill2008performance}. We need at least one bubble in the ring, regardless of how many stages constitute it. When the number of data items in the ring is small, the other stages of the pipeline will be filled with \emph{holes}. When there is more holes than data, the pipeline is said to be \emph{data-limited}; when there is more data than holes, it is said to be \emph{hole-limited} \cite{gill2008performance}. These operation modes correspond to the token-limited and bubble-limited modes, respectively, in the linear Muller pipeline.

We use our tool to study the effect of varying the ring occupancy, by building the ring with a different number of stages and changing the token count. As this is a Muller pipeline, each C-element needs alternating input sequences to be able to transition from $0$ to $1$ every stage.

In order to keep the pipeline running and avoid deadlock, the process of correctly initializing the stages of the ring is crucial. As previously mentioned, there must be at least one bubble in the pipeline. For each combination of tokens and stages, we calculate the number of bubbles needed and we fill the pipeline in the following manner:
\begin{itemize}
    \item If the number of bubbles is much larger than the number of tokens, we start by filling the pipeline with bubbles, and insert tokens equally paced from one another.
    \item If the number of bubbles is much lower than the number of tokens, we start by inserting tokens, and spread the bubbles in between.
    \item If there is only one bubble, it doesn't matter where it is inserted. Same if there is only one token.
    \item A token is always inserted as a data token and a spacer that are not separated by a bubble.
\end{itemize}

The results for these settings are shown in Figure~\ref{fig:muller-ring-sweepTokens}. The first point of each line (from the left) represents the maximum number of tokens allowed for the corresponding number of stages (recall that this count represents, in fact, a data token and a spacer). The top left region represents the bubble-limited operation mode, where one can clearly see that $P(\text{fail})$ gets higher. 
The increasing number of stages also seems to play a role in this trend. From what we have previously observed 
 we can conjecture that this is because an idle (waiting) stage (MCE) has the highest fault probability, while one that processes a token/transition is more resilient. By adding stages while keeping the number of tokens constant, we add idle stages -- consequently $P(\text{fail})$ increases.

As we move to the edges of the token-limited region where the number of bubbles largely exceeds the number of tokens,  $P(\text{fail})$ converges to a steady percentage of approximately 45\%.

\begin{figure}
    \centering
    \includegraphics[width=0.7\columnwidth]{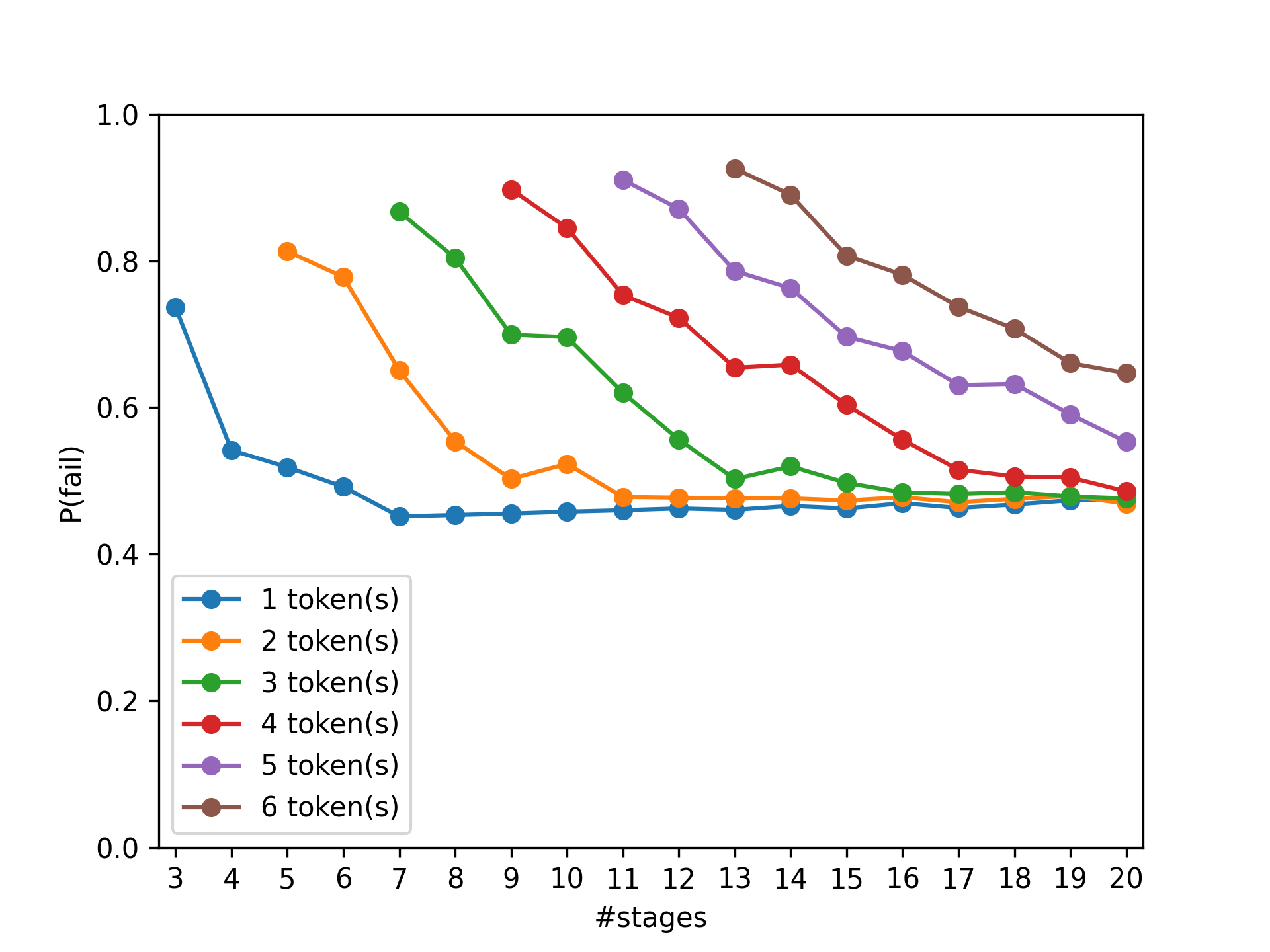}
    \caption{Influence of number of tokens and stages on the probability to fail $P(\text{fail})$. Ring pipeline with varying number of stages and tokens. Delays as follows: 1 (INV), 5 (MCE). $T=500$. A token encompasses a data token along with a spacer, so 1 token means one of each, following the 4-phase communication protocol. In the case of a 2-phase protocol, this token count is doubled.}
    \label{fig:muller-ring-sweepTokens}
\end{figure}

Finally, we compare throughput and probability to fail as a function of the same ring pipeline with a varying number of (4-phase) tokens.
It has been previously observed \cite{williams1987self,gill2008performance} that the throughput as a function of tokens behaves as a canopy plot: it is low for small numbers of tokens (token-limited), high in the middle, and low for large numbers of tokens (bubble-limited).
Figure \ref{fig:canopy} compares this behavior with the failure probability as determined by our tool for execution prefixes of length $T = 200$.
While the canopy diagram suggests that 4 and 5 tokens yield optimum throughput, the failure probability favors a lower token count. So for maximum performance the better choice would be 4 tokens.
This result is not general and seems to depend on the design choices, but the general strategy should be to also consider $P(\text{fail})$ in the system design.

\begin{figure}
    \centering
    \includegraphics[width=0.7\columnwidth]{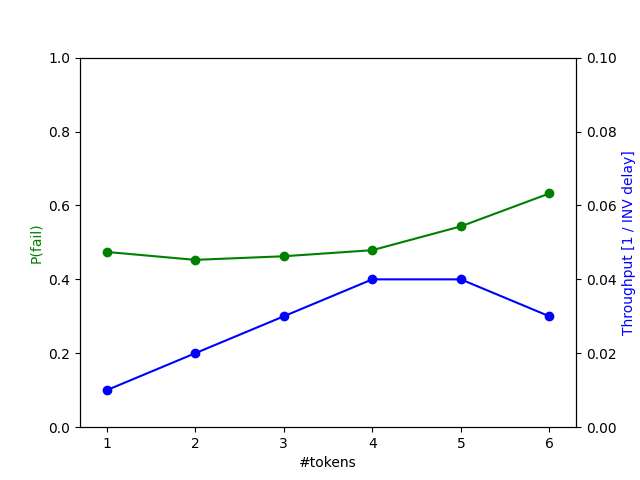}
    \caption{Influence of number of tokens on the probability to fail $P(\text{fail})$ and the throughput (in 4-phase tokens per {\tt INV} delay passing {\tt MCE1}). Ring pipeline with 20 stages and varying number of tokens. Delays as follows: 1 (INV), 5 (MCE). $T=200$. A token encompasses a data token along with a spacer (4-phase; for 2-phase interpretation multiply by factor 2).}
    \label{fig:canopy}
\end{figure}

\subsection{Multi-bit QDI designs}\label{sec:multi}

To demonstrate the ability of our tool to handle larger, multi-bit designs,
  we ran the algorithm to determine susceptible windows on execution prefixes
  of duration $T=500$.
As circuits under test we used 4-bit and 8-bit versions of the previously analyzed
  linear and ring pipeline.
Our tool reported the following values for $P(\text{fail})$:
(i) the 3-stage linear pipeline with 4-bit resulted in $0.22$ and with 8-bit in $0.10$.
(ii) the 3-stage ring pipeline with 4-bit resulted in $0.22$ and with 8-bit in $0.16$.
The observed decrease of $P(\text{fail})$ for higher pipeline width is as expected from literature: while the last rail to switch is the most critical one (it triggers the completion detector), the faster bits are less critical and hence contribute to lowering the overall fault probability -- with growing impact for increasing bit number.

All results were obtained within minutes on a MacBook Pro (M2, 2022) with 24~GB RAM.

\end{document}